\documentclass[a4paper, 11pt]{article}
\usepackage{diagbox}
\usepackage{epstopdf}
\usepackage{stmaryrd}
\usepackage{epsfig}
\usepackage[notref,notcite]{showkeys}
\usepackage{mathtools}
\usepackage{showkeys}
\usepackage{graphicx}
\usepackage[normalem]{ulem}
\usepackage{cancel}
\usepackage{hyperref}
\usepackage{amsmath,amsthm,amssymb,amscd}
\usepackage{wasysym}
\usepackage{tikz}
\usepackage{multirow}
\usepackage{ulem}
\usepackage{dsfont}
\usepackage{array}

\DeclareMathOperator{\Tr}{Tr}
\usepackage{color}

\newcommand{\rd}{{\mathrm d}}
\newcommand{\R}{\mathbb R}
\newcommand{\C}{\mathbb C}

\newcommand{\UNCD}{{\mathrm{UNCD}}}

\newcommand{\vertiii}[1]{{\left\vert\kern-0.25ex\left\vert\kern-0.25ex\left\vert #1 
    \right\vert\kern-0.25ex\right\vert\kern-0.25ex\right\vert}}

\newcommand{\Acal}{{\mathcal A}}
\newcommand{\Bcal}{{\mathcal B}}

\newcommand{\Hcal}{{\mathcal H}}

\newcommand{\PiAcal}{{\Pi_{\Acal}}}
\newcommand{\PiBcal}{{\Pi_{\Bcal}}}
\newcommand{\Ker}{{\mathrm{Ker}}}
\newcommand{\Ima}{{\mathrm{Im}}}

\newcommand{\N}{\mathbb{N}}

\newcommand{\bbone}{\mathbb{I}}

\newcommand{\nab}{n_{\Acal,\Bcal}}
\newcommand{\nabmin}{\nab^{\min}}
\newcommand{\na}{n_{\Acal}}
\newcommand{\nb}{n_{\Bcal}}
\newcommand{\mab}{m_{\Acal,\Bcal}}
\newcommand{\Mab}{M_{\Acal, \Bcal}}

\newcommand{\Zc}{Z_{\textrm c}}
\newcommand{\Zr}{Z_{\textrm r}}

\newtheorem{theorem}{Theorem}
\newtheorem{lemma}[theorem]{Lemma}
\newtheorem{definition}[theorem]{Definition}

\newtheorem{proposition}[theorem]{Proposition}

\makeatletter
\newcommand{\oset}[3][0ex]{%
	\mathrel{\mathop{#3}\limits^{
			\vbox to#1{\kern-11\ex@
				\hbox{$\scriptstyle#2$}\vss}}}}
\makeatother

\begin{document}

\title{ Relating incompatibility, noncommutativity, uncertainty and\\
Kirkwood-Dirac nonclassicality}

\author{Stephan De Bi\`evre\thanks{Stephan.De-Bievre@univ-lille.fr}\\
Univ. Lille, CNRS, Inria, UMR 8524 \\ Laboratoire Paul Painlev\'e, F-59000 Lille, France
}


\date{\today}

\maketitle
\begin{abstract}
We provide an in-depth study of the recently introduced notion of completely incompatible observables and its links to the support uncertainty and to the Kirkwood-Dirac nonclassicality of pure quantum states. The latter notion has recently been proven central to a number of issues in quantum information theory and quantum metrology. In this last context, it was shown that a quantum advantage requires the use of Kirkwood-Dirac nonclassical states.
We establish sharp bounds of very general validity that imply that the support uncertainty is an efficient Kirkwood-Dirac nonclassicality witness.  When adapted to completely incompatible  observables that are close to mutually unbiased ones, this bound allows us to fully characterize the Kirkwood-Dirac classical states as the eigenvectors of the two observables. We show furthermore that complete incompatibility implies several weaker notions of incompatibility, among which features a strong form of noncommutativity.

\end{abstract}

\section{Introduction} 
Among the salient characteristics of quantum mechanics, distinguishing it from classical mechanics,  feature prominently the incompatibility and noncommutativity of two observables and the associated uncertainty principles. They are, together with entanglement, the main ingredients for the study of the classical-quantum transition within the quantum state space. The subject continues to attract considerable attention from the viewpoint of foundational issues~\cite{Jo07, Sp08,  HeWo10,  Pu14, Dr15,  HeEtAl16, YuSwDr18, Lo18,   CaHeTo20, UoEtAl21} as well as in the context of the search for a quantum advantage in various protocols of quantum information and metrology~\cite{Fe11, LuBa12, BaLu14, ThGiChHoBaLu16,  ArEtAl20, ArDrHa21,DeFaKa19, MoKa21}. In particular, the use of Kirkwood-Dirac nonclassical states has recently been proven to be essential in the latter context~\cite{ArEtAl20, ArDrHa21}.

It was pointed out in~\cite{SDB21} that when incompatibility is equated with noncommutativity (as is often the case~\cite{HeWo10, HeEtAl16,  DeFaKa19,  CaHeTo20, MoKa21, UoEtAl21}), only a very weak notion of incompatibility is obtained and a much stronger notion, referred to as ``complete incompatibility'' was proposed. (See Definition~\ref{def:COINC} below.) The latter provides a mathematical expression to the physical idea that the measurement of  a second observable after the measurement of a first one always perturbs the result obtained in the first measurement, whatever the pre-measurement state. Complete incompatibility was then shown to lead to a strong uncertainty relation for all pure states which in turn was proven to be linked to the notion of Kirkwood-Dirac nonclassicality. It is the goal of this paper to further explore the consequences of the complete incompatibility of two observables and its relation to various weaker notions of incompatibility (Proposition~\ref{prop:incompatibilities}), as well as to deepen the  logical links between this notion, the uncertainty inherent in quantum states, and their Kirkwood-Dirac nonclassicality. One of the main results of our analysis is  that when two observables are completely incompatible and -- in a sense we make precise -- close to mutually unbiased, then all states, except the eigenstates of the two observables, are Kirkwood-Dirac nonclassical (Theorem~\ref{thm:COINC_KDNC}). 

When two observables do not satisfy the stringent complete incompatibility condition, our results still permit to partially characterize the KD-nonclassical states (Theorem~\ref{thm:NCbound} and Theorem~\ref{thm:pert_mubclassical}). 

We shall work in the context of quantum mechanics on a Hilbert space $\Hcal$ of finite dimension $d$ and formulate our definitions and results in terms of two orthonormal bases $\Acal=\{|a_i\rangle| 1\leq i\leq d\}$ and $\Bcal=\{|b_j\rangle| 1\leq j\leq d\}$ that one can -- but need not -- think of as the eigenbases of two observables $A$ and $B$. We start by showing (Proposition~\ref{prop:incompatibilities}) that the complete incompatibility of two bases/observables implies a number of weaker forms of incompatibility, including noncommutativity,  that we each interpret physically in terms of successive repeated measurements.     
We show in particular that when two observables are completely incompatible, no spectral projector of $A$ commutes with any spectral projector of $B$. This is much stronger than saying that the two observables do not commute, which only implies that at least one spectral projector of $A$ does not commute with one of $B$. Nevertheless, as we show, even this very strong form of noncommutativity does not imply complete incompatibility, which is a stronger notion still. 

To study the link of complete incompatibility with uncertainty, we proceed as follows. Given a state $\psi\in\Hcal$, we define $\na(\psi)$ (respectively $\nb(\psi)$) to be the number of nonvanishing components of $\psi$ on $\Acal$ (respectively $\Bcal$):
$$
\na(\psi)=\sharp\{i\in\llbracket 1,d\rrbracket \mid \langle a_i|\psi\rangle\not=0\},\quad \nb(\psi)=\sharp\{j\in\llbracket 1,d\rrbracket \mid \langle b_j|\psi\rangle\not=0\}.
$$
We then introduce the uncertainty diagram of the pair $(\Acal, \Bcal)$ to be the set of points $(\na,\nb)\in\llbracket 1, d\rrbracket^2$ in the $\na\nb$-plane for which there exists $\psi\in\Hcal$ so that $\na(\psi)=\na$, $\nb(\psi)=\nb$. Whereas for an arbitrary choice of bases $\Acal$ and $\Bcal$ the uncertainty diagram can have a complex structure, we show it can be easily characterized in the case the two bases are completely incompatible. Indeed, in that case, it is composed of \emph{all} points $(\na,\nb)\in\llbracket 1, d\rrbracket^2$ for which $\na+\nb\geq d+1$ (Theorem~\ref{thm:COINC_KDNC}). This last inequality, which is in fact equivalent to complete incompatibility, can be interpreted as an uncertainty relation: it says $\na(\psi)$ and $\nb(\psi)$ cannot both be small. 

To illustrate these first results, we analyse the case where $\Acal$ and $\Bcal$ are mutually unbiased bases (MUB), meaning $|\langle a_i|b_j\rangle|=\frac1{\sqrt d}$, for all $i,j$. Those have attracted considerable attention over the years~\cite{Schw60, Iv81, PlRoPe06, DuEnBeZy10, FaKa19} and many open questions concerning their classification remain in dimensions $d\geq 6$~\cite{Ba22}.
The MUB are sometimes referred to as ``complementary''  or  ``maximally incompatible''~\cite{DuEnBeZy10} because if a measurement in $\Acal$ yields the outcome $|a_i\rangle$, then a subsequent measurement in $\Bcal$ yields any of the outcomes $|b_j\rangle$ with the same probability.  It turns out, nevertheless, that not all MUB are completely incompatible, as already pointed out in~\cite{SDB21}. Whereas, in dimensions $d=2,3,5$ this is easily seen to be the case, no MUB are completely incompatible when $d=4$~\cite{SDB21}. In higher dimension, the situation is unclear: it is in particular not known in which dimension there exist MUB that are completely incompatible. We show in this paper (Theorem~\ref{thm:COINC_mab}) that, in all dimensions, there exist completely incompatible bases $\Acal, \Bcal$ for which $\mab:=\min_{i,j}|\langle a_i|b_j\rangle|\geq \frac1{\sqrt d}-\delta$, for any $\delta>0$. In other words, there exist completely incompatible bases that are arbitrary close to being mutually unbiased.

Having explained that the complete incompatibility of two observables is equivalent to a strong uncertainty relation for all states $|\psi\rangle\in\Hcal$, we turn to the central question of this paper, which are the links between  complete incompatibility, uncertainty and ``nonclassicality''. For that purpose, we use the notion of nonclassicality associated to the Kirkwood-Dirac distribution. 
Recall that, given $\Hcal$ and two orthonormal bases $\Acal=\{|a_i\rangle\}$ and $\Bcal=\{|b_j\rangle\}$, the Kirkwood-Dirac (KD) distribution of a state $\psi$~\cite{Ki33, Di45},
\begin{equation}\label{eq:KD}
Q(\psi)_{ij}=\langle a_i|\psi\rangle\langle\psi|b_j\rangle \langle b_j|a_i\rangle, 1\leq i,j\leq d,
\end{equation}
is a quasi-probability distribution  somewhat similar in spirit to the Wigner distribution or Glauber-Sudarshan $P$-function~\cite{cagl69a, cagl69b} in continuous variable quantum mechanics, which are intimately linked to the choice of two conjugate quadratures $X$ and $P$.  
It is complex-valued and satisfies
$
\sum_{ij} Q(\psi)_{ij}=1,
$
with marginals
$
\sum_j Q(\psi)_{ij}=|\langle a_i|\psi\rangle|^2,  \sum_i Q(\psi)_{ij}=|\langle b_i|\psi\rangle|^2.
$
A state $\psi\in\Hcal$ is said to be \emph{KD classical} if its KD distribution is real nonnegative everywhere; in other words, if its KD distribution is a probability distribution. Otherwise it is \emph{KD nonclassical}.  KD nonclassicality  is used  in quantum tomography~\cite{LuBa12, BaLu14, ThGiChHoBaLu16} as well as in the theory and applications of weak measurements, (non)contextuality, and their relation to nonclassical effects in quantum mechanics~\cite{Sp08, Dr15, Pu14, Lo18}. Also, KD nonclassicality has been shown to be an essential ingredient necessary to obtain an operational quantum advantage in postselected metrology~\cite{ArEtAl20, ArDrHa21}. 

For these reasons, the question of how to characterize the KD-nonclassical states poses itself naturally. One way to do this is to identify a nonclassicality witness. By this, we mean a quantity depending on the state $|\psi\rangle$ that can be computed or measured without determining the KD-distribution of the state, and whose value can reveal its KD nonclassicality.  It is crucial to keep in mind that the KD distribution and hence the KD nonclassicality of $|\psi\rangle$  depend not only on $|\psi\rangle$, but also on the two bases $\Acal, \Bcal$ used to define it, the choice of which is in applications linked to the identification of two observables $A$ and $B$ of which they are eigenbases. The questions that arise are therefore: under what conditions on the bases $\Acal$ and $\Bcal$ do there exist KD-nonclassical states, how many are there, where in Hilbert space can they be found and how strongly nonclassical are they? 

A form of incompatibility or of noncommutativity is needed between the projectors $|a_i\rangle\langle a_i|$ and $|b_j\rangle\langle b_j|$ for KD-nonclassical states to exist. Indeed, if for example $A$ and $B$ have nondegenerate spectra, and if they are compatible in the usual sense that  $[A,B]=0$, then all those projectors commute, and therefore each $|a_i\rangle$ is up to a phase equal to some $|b_j\rangle$. It is then immediate from~\eqref{eq:KD} that $Q_{ij}(\psi)\geq 0$ for all $\psi$ so that there are no KD-nonclassical states in $\Hcal$.  Under the assumption that $A$ and $B$ are incompatible in the sense that they do not commute, a sufficient but non-optimal condition for a state to be KD nonclassical was given in~\cite{ArDrHa21}. (See Theorem~\ref{thm:arvid} below.) In the present paper, we sharpen and optimize that result in several ways. In particular, our main results under the assumption that the bases are completely incompatible are summarized in the following theorem. It is a corollary of Theorem~\ref{thm:COINCsupport_new} and Theorem~\ref{thm:pert_mubclassical}. 

\begin{theorem} \label{thm:COINC_KDNC}
Two bases $\Acal$ and $\Bcal$ are COINC if and only if for all $\psi\in\Hcal$, $\na(\psi)+\nb(\psi)\geq d+1$.\\  In that case: 
\begin{itemize}
\item[(i)] For all $1\leq \na,\nb\leq d$ with $\na+\nb\geq d+1$, $\exists \psi\in\Hcal$, $\na(\psi)=\na, \nb(\psi)=\nb$.
\item[(ii)] If $\psi\in\Hcal$ satisfies $\nab(\psi):=\na(\psi)+\nb(\psi)>d+1$, then $\psi$ is KD nonclassical.
\item[(iii)] If a state $\psi\in\Hcal$ is KD classical, then $\nab(\psi):=\na(\psi)+\nb(\psi)=d+1$.
\item[(iv)] If in addition $\Acal$, $\Bcal$ are such that
\begin{equation}\label{eq:mablowerbound}
\mab^2:=\min_{ij}|\langle a_i|b_j\rangle|^2> \frac1{d}\left(\frac{\frac{d(d-1)}{2}}{1+\frac{d(d-1)}{2}}\right),
\end{equation}
and hence in particular when they are mutually unbiased, then the only classical states are the vectors in $\Acal$ and $\Bcal$. 
\end{itemize}
\end{theorem}
The first part of this result can be suggestively paraphrased as saying that, when $\Acal$ and $\Bcal$ are completely incompatible, all states, except possible those that have minimal support uncertainty, are KD nonclassical. Part (ii) in particular shows that the \emph{support uncertainty} $\nab(\psi):=\na(\psi)+\nb(\psi)$ is a KD-nonclassicality witness. According to (iv),  if in addition $\mab$ is close to its maximal possible value $1/\sqrt d$, only the elements of the two bases $\Acal$ and $\Bcal$ are classical.

Since complete incompatibility is a stringent condition on the bases $\Acal$ and $\Bcal$, not necessarily satisfied in all situations of interest, it is important to understand what can be said about the KD-nonclassical states in such situations as well.  It was shown in~\cite{SDB21} that, if the bases satisfy the weaker incompatibility condition that 
$$
\mab:=\min_{i,j} |\langle a_i|b_j\rangle|>0,
$$
 then part~(ii) and~(iii)  of Theorem~\ref{thm:COINC_KDNC} still hold (but not~(i)). In other words, under that condition, the support uncertainty is still a KD-nonclassicality witness. We show here this remains true in much greater generality, even if $\mab=0$, but provided not too many of the overlaps $\langle a_i|b_j\rangle$ vanish (Theorem~\ref{thm:NCbound}). We further establish a sharpening of the bound in (ii) of the above theorem:  provided $\mab$ is sufficiently close to $1/\sqrt d$, any $|\psi\rangle$ that is not an eigenvector of one of the two observables and satisfies $n_{\Acal,\Bcal}(\psi)>d$ is nonclassical (Theorem~\ref{thm:pert_mubclassical}).

The paper is organized as follows. In Section~\ref{s:setup} we define the mathematical framework within which we shall work. Section~\ref{s:compatibility} briefly reviews the standard notions of incompatibility that appear in the literature and their link with noncommutativity. Section~\ref{s:uncertainties} similarly reviews the various uncertainty relations available in the literature and introduces the less well known support uncertainty that we shall consider here, as well as a very general support uncertainty relation that it satisfies.  This sets the stage for our definition and characterization of complete incompatibility in Section~\ref{s:COINC} and for the explanation of its link with a strengthened support uncertainty relation. A study of the relationships between  mutual unbiasedness and complete incompatibility is proposed in Section~\ref{s:COINC_MUB}. In Section~\ref{s:supuncwitness} we establish under which conditions the support uncertainty of a state can serve as a KD nonclassicality witness. This allows us to provide a complete characterization of all KD-classical states under suitable conditions on $\Acal$ and $\Bcal$.  Some final conclusions are drawn in Section~\ref{s:conclusion}.

 \section{Setup}\label{s:setup}
In this preliminary section, we describe the basic objects entering our study and fix the notation we shall use.  Throughout, we will consider two distinct orthonormal bases $\Acal=\{|a_i\rangle| 1\leq i\leq d\}$ and $\Bcal=\{|b_j\rangle| 1\leq j\leq d\}$ of a $d$-dimensional Hilbert space $\Hcal$. 
Given two such bases,
we define their transition matrix $U_{ij}=\langle a_i|b_j\rangle$. 
Introducing 
$$
\mab:=\min_{i,j}|\langle a_i| b_j\rangle|,\quad \Mab=\max_{i,j}|\langle a_i| b_j\rangle|,
$$
one has
\begin{equation}\label{eq:mMbounds}
0\leq \mab\leq d^{-1/2}\leq \Mab\leq 1,
\end{equation}
which follows immediately from the fact that the transition matrix is unitary so that 
 $$
d=\Tr \bbone =\sum_{i,j} |\langle b_j|a_i\rangle|^2.
 $$
 Clearly
 \begin{equation}\label{eq:mabMab}
 \mab>0\Rightarrow \Mab<1.
 \end{equation}
 The condition $\Mab<1$ means $\Acal$ and $\Bcal$ have no basis vectors in common. We will see below (Proposition~\ref{prop:incompatibilities}) that both the conditions $\Mab<1$ and $\mab>0$ can be viewed as a form of incompatibility between the two bases, with $\mab>0$ being the stronger one of the two.

Considering a column of $U$ containing a term for which $\Mab=|\langle a_i|b_j\rangle|$, one has
\begin{equation}\label{eq:mMbound1}
1\geq \Mab^2 +\mab^2(d-1)
\end{equation}
and similarly, considering a column of $U$ containing a term for which $\mab=|\langle a_i|b_j\rangle|$, one has
\begin{equation}\label{eq:mMbound2}
1\leq \mab^2 +\Mab^2(d-1).
\end{equation}
Hence
$$
\mab=\frac{1}{\sqrt d}\Leftrightarrow \Mab=\frac{1}{\sqrt d}\Leftrightarrow \mab=\Mab;
$$
any of these conditions characterizes mutually unbiased bases.

When two bases $\Acal$ and $\Bcal$ are specified, we can associate to each $\psi\in\Hcal$ its $\Acal$- and $\Bcal$-\emph{representations}, which are the 
vectors  of their components $u(\psi)=(\langle a_1|\psi\rangle,\dots,\langle a_d|\psi\rangle)\in\C^d$ on the $\Acal$ basis and $v(\psi)=(\langle b_1|\psi\rangle,\dots,\langle b_d|\psi\rangle)\in\C^d$ on the $\Bcal$ basis. One has
\begin{equation}\label{eq:ABreps}
u(\psi)=Uv(\psi)\quad\textrm{and}\quad
\sum_{i=1}^d |u_i(\psi)|^2=\langle \psi|\psi\rangle=\sum_{j=1}^d |v_j(\psi)|^2.
\end{equation}
Given a state $\psi$, so that $\langle \psi|\psi\rangle=1$, $|u_i(\psi)|^2$ and $|v_j(\psi)|^2$ form two probabilities representing uncertainties which can be measured in various ways. We will come back to this extensively below and in particular study the joint uncertainty with respect to both bases. 
 
Given  bases $\Acal$ and $\Bcal$, we introduce a family of orthogonal projectors as follows. For every $S, T\subset \llbracket 1,d\rrbracket:=\{1,2,\dots, d\}$, 
$$
\PiAcal(S)=\sum_{i\in S} |a_i\rangle\langle a_i|,\quad \PiBcal(T)=\sum_{j\in T} |b_j\rangle\langle b_j|. 
$$
In some cases the two bases $\Acal, \Bcal$ arise as the eigenbases of two self-adjoint operators $A$ and $B$ on $\Hcal$. If their eigenvalues are non-degenerate, then the corresponding bases are unique (up to a global phase and possible reordering of the basis vectors), but not otherwise. In what follows, given two observables $A$ and $B$, we shall always assume a particular eigenbasis has been chosen for each. In many examples of interest, however, the two bases $\Acal$ and $\Bcal$ are not associated to observables, as we shall see. 
If a basis is given, we will say a self-adjoint operator $A$ with eigenvalues $\alpha_1<\alpha_2<\dots <\alpha_L$ is adapted to the basis $\Acal$ if  the latter is a basis of eigenvectors for $A$. We can associate a projective partition of unity to such an adapted observable by considering the partition of $\llbracket 1,d\rrbracket$ defined by 
\begin{equation}\label{eq:eigenvalueset}
S_\ell=\{i\in \llbracket 1, d\rrbracket \mid A|a_i\rangle=\alpha_\ell |a_i\rangle\}, \quad 1\leq \ell\leq L.
\end{equation}
The $\PiAcal(S_\ell)$ are then the eigenprojectors of $A$ and one has
$$
\cup_\ell S_\ell=\llbracket 1,d\rrbracket, \quad \sum_\ell \PiAcal(S_\ell)=\bbone.
$$ Finally, there is a unique projection-valued measure associated to $A$ in the usual way. Let $I\subset \R$. Then we introduce the spectral projector of $A$ for $I$ to be the operator $\Pi(A\in I)$ defined as follows:
$$
\Pi(A\in I)=\sum_{\alpha_\ell\in I}\PiAcal(S_\ell).
$$
This is the sum of the projectors onto the basis vectors whose corresponding eigenvalue is in $I$. 
Note that these spectral projectors are all of the type $\PiAcal(S)$ for some $S\subset\llbracket 1,d\rrbracket$. 

These considerations lead to the following description of quantum measurements associated to a basis $\Acal$, that we need for our discussion of  both uncertainty and incompatibility.
We only need some basic elements of the quantum theory of projective measurements as discussed in in~\cite{Jo07, NiChu10, CTDL15}, that we briefly recall. 

We  in particular only consider selective measurements: the outcome of an individual measurement is always recorded. Given a basis $\Acal$, one can associate to every  partition $S_1,\dots, S_L$ of the index set $\llbracket 1,d\rrbracket$ a family of projectors $\PiAcal(S_1), \dots\PiAcal(S_L)$ forming a projective partition of unity in the sense that
$$
\sum_{\ell=1}^L\PiAcal(S_\ell)=\bbone.
$$ 
To such a family corresponds a selective projective measurement the outcomes of which are the $S_\ell$. In the case where the $S_\ell$ are as defined in~\eqref{eq:eigenvalueset}, this is commonly referred to as a measurement of $A$. The outcomes $S_\ell$ are then identified with the eigenvalues $\alpha_\ell$.   If initially the system is in the state $\psi\in\Hcal$ and the outcome $S_\ell$ is realized then the post-measurement state is the pure state $\PiAcal(S_\ell)\psi/\|\PiAcal(S_\ell)\psi\|$.  The probability of this outcome is $\|\PiAcal(S_\ell)\psi\|^2$. This is what we shall refer to as a measurement in the basis $\Acal$.  When $L=d$ and $S_\ell=\{\ell\}$, for $\ell=1,\dots, d$,  we will say the measurement is finegrained. Otherwise it is coarsegrained. When $L=2$, a case we will frequently consider, the corresponding measurement has only two possible outcomes and one has  $S_1=S, S_2=S^{\textrm{c}}$, for some $S\subset\llbracket 1,d\rrbracket$. In line with the preceding terminology, we can also refer to this as a measurement of $\PiAcal(S)$. Indeed, the projector 
$\PiAcal(S)$ has two eigenvalues, $1$ and $0$, that correspond to the outcomes $S$ and $S^{\textrm{c}}$. 

Note that the transition matrix $U$ depends on the order in which the basis vectors are listed. None of the definitions or results in this work will depend on that order and we will freely reorder the bases when it is convenient. Similarly, the results only depend on the basis vectors up to a global phase. In particular, if $k$ of the basis vectors of $\Acal$ coincide with $k$ of those of $\Bcal$ then one can consider $|a_i\rangle=|b_i\rangle$ for $1\leq i\leq k$, after a possible reordering and adjusting of the phases. From the point of view of incompatibility, uncertainty and KD nonclassicality, nothing interesting then happens on the subspace of $\Hcal$ spanned by the first $k$ basis vectors and the entire analysis in this paper can then be done on the $d-k$-dimensional complementary subspace. In particular, one then has
$$
Q_{ij}(\psi)=|\langle \psi|a_i\rangle|^2\delta_{ij}, \quad \forall 1\leq i,j\leq k.
$$
More generally, when $U$ is block-diagonal, then so is $Q(\psi)$ and the analysis should be performed for each block separately. 

For later use, we introduce the following notation. For $S, T\subset \llbracket 1, d\rrbracket$ we introduce the  $|S|\times |T|$ matrix $U(S,T)$ with entries 
\begin{equation}\label{eq:UST}
U(S,T)_{k\ell}=\langle a_{i_k}|b_{j_\ell}\rangle,
\end{equation}
 where $1\leq k\leq |S|, 1\leq \ell\leq |T|$, and $i_k\in S, j_\ell\in T$. Note that Rank $U(S,T)\leq\min \{|S|, |T|\}$. Also, in what follows
 \begin{equation}\label{eq:hcalST}
\Hcal(S,T)=\PiAcal(S)\Hcal\cap\PiBcal(T)\Hcal.
\end{equation}

\section{(In)compatibility and (non)commutativity  }\label{s:compatibility}
To study (in)compatibility, we need to discuss repeated successive measurements in the bases $\Acal$ and $\Bcal$, a subject we turn to in this section. 
We start with a brief but critical analysis of the standard relationship between compatibility and commutativity of observables~\cite{Lu51, CTDL15}, which is helpful in motivating the notion of complete incompatibility introduced in  Section~\ref{s:COINC}.

 Suppose on a state $\psi$  successive measurements in $\Acal$, then $\Bcal$, yield the outcomes $S$, then $T$. 
We will say these outcomes are compatible if, when measuring again in $\Acal$ after having obtained $T$, the outcome $S$ occurs with probability one. This means that the measurement of $\Bcal$ after the first measurement of $\Acal$ has not disturbed the first outcome.  Repeating measurements in $\Acal$ and/or $\Bcal$ will then systematically yield the same outcomes, always with probability one. As a result, we adopt the following definition:
\begin{definition}\label{def:compatibleoutcomes}  $S,T$ are compatible outcomes for a measurement in $\Acal$ followed by one in $\Bcal$ if
\begin{equation}\label{eq:compatible}
\PiBcal(T)\PiAcal(S)\Hcal\subset \PiAcal(S)\Hcal.
\end{equation}
\end{definition}
Indeed, if the pre-measurement state is $\psi$, the (nonnormalized) state after the first two measurements is $\PiBcal(T)\PiAcal(S)\psi$.  For the second measurement in $\Acal$ to yield $S$ with probability one  for all such pre-measurement states $\psi$, one needs that  $\PiBcal(T)\PiAcal(S)\psi\in\PiAcal(S)\Hcal$ for all such $\psi$. This justifies the definition. For brevity, we will simply say $S$ and $T$ are compatible. 

Eq.~\eqref{eq:compatible} is equivalent to  $[\PiAcal(S),\PiBcal(T)]=0$:
\begin{equation}\label{eq:compatibleST}
\PiBcal(T)\PiAcal(S)\Hcal\subset \PiAcal(S)\Hcal\Leftrightarrow [\PiAcal(S),\PiBcal(T)]=0.
\end{equation}
 Indeed, Eq.~\eqref{eq:compatible} means $\PiBcal(T)$ leaves $\PiAcal(S)\Hcal$ invariant, which implies it also leaves its orthogonal complement $\PiAcal(S^{\textrm{c}})\Hcal$ invariant so that $\PiAcal(S)\PiBcal(T)\PiAcal(S^\textrm{c})=0$. Hence
\begin{eqnarray*}
\PiBcal(T)\PiAcal(S)\psi&=&\PiAcal(S)\PiBcal(T)\PiAcal(S)\psi=\PiAcal(S)\PiBcal(T)\psi-\PiAcal(S)\PiBcal(T)\PiAcal(S^\textrm{c})\psi\\
&=&\PiAcal(S)\PiBcal(T)\psi.
\end{eqnarray*}
The converse implication is obvious. One then also has
\begin{equation}\label{eq:compSTcomm}
\PiBcal(T)\PiAcal(S)\Hcal\subset \PiAcal(S)\Hcal\Leftrightarrow [\PiAcal(S),\PiBcal(T)]=0\Leftrightarrow \PiAcal(S)\PiBcal(T)\Hcal\subset \PiBcal(T)\Hcal.
\end{equation}
As a result, if first a measurement in $\Bcal$ is made, and the outcome $T$ obtained, then a subsequent measurement of $\Acal$ with outcome $S$ does not perturb the first measurement. In other words, if $S,T$ are compatible outcomes for a measurement in $\Acal$ followed by one in $\Bcal$, then $T,S$ are compatible outcomes for a measurement in $\Bcal$ followed by one in $\Acal$.  

We will say two observables $A$ and $B$ are compatible when all outcomes associated to their eigenvalues $\alpha_\ell, \beta_k$ (see Eq.~\eqref{eq:eigenvalueset}) are compatible in the sense of Eq.~\eqref{eq:compatible}. In view of what precedes, this means $A$ and $B$ are compatible if and only if they commute. Mathematically, this is a very strong, restrictive property. It entails a number of further physical properties such as joint measurability and measurement non-disturbance~\cite{HeWo10, HeEtAl16, DeFaKa19}.

Note that $S$ and $T$ are mutually exclusive when, after a measurement in $\Acal$ yielded the outcome $S$, a measurement in $\Bcal$ cannot yield the outcome $T$. This is equivalent to $\PiBcal(T)\PiAcal(S)\Hcal=0$. This, in turn, implies Eq.~\eqref{eq:compatible}, so that, in this sense, mutually exclusive events are compatible. While this sounds  odd, considering the standard everyday meaning of the words ``compatible'' and ``mutually exclusive'', it is implicit in  all discussions on quantum mechanics and we shall also adopt this practice. Indeed, it often occurs that two observables commute, while many outcomes $\alpha_\ell, \beta_k$ are mutually exclusive.

What happens if $[\PiAcal(S), \PiBcal(T)]\not=0$?  Suppose a measurement in $\Acal$ on the pre-measurement state $\psi$ yields the outcome $S$ so that the post-measurement state is, up to normalization, $\PiAcal(S)\psi$. Suppose that a subsequent measurement in $B$ yields the outcome $T$, and the post-measurement state is $\PiBcal(T)\PiAcal(S)\psi$. As we saw, this second measurement has not perturbed the first 
if and only if 
\begin{equation}\label{eq:compatiblebis}
\PiAcal(S)\PiBcal(T)\PiAcal(S)\psi=\PiBcal(T)\PiAcal(S)\psi,
\end{equation}
which is equivalent to 
$\PiBcal(T)\PiAcal(S)\psi\in\PiAcal(S)\Hcal$. Do such states exist? The answer is yes, if and only if
\begin{equation}\label{eq:nonperturb}
\PiAcal(S)\Hcal\cap\PiBcal(T)\Hcal\not=\{0\}.
\end{equation}
 In that case, taking $0\not=\phi\in\PiAcal(S)\Hcal\cap\PiBcal(T)\Hcal$, $\phi'\in\PiAcal(S^{\textrm c})\Hcal$, and setting $\psi=\phi+\phi'$, one has indeed~Eq.~\eqref{eq:compatiblebis}.  Note that Eq.~\eqref{eq:compatiblebis} is equivalent to the statement
$$
\PiAcal(S)\psi\in \Ker[\PiAcal(S),\PiBcal(T)].
$$
Of course, even though $[\PiAcal(S), \PiBcal(T)]\not=0$, generally the kernel of $[\PiAcal(S), \PiBcal(T)]$ is not trivial. In fact, it is always true that
$$
\PiAcal(S)\Hcal\cap\PiBcal(T)\Hcal\subset \Ker[\PiAcal(S), \PiBcal(T)].
$$
In conclusion, if $\PiAcal(S)\Hcal\cap\PiBcal(T)\Hcal\not=\{0\}$, then, for any state $\psi\in \PiAcal(S)\Hcal\cap\PiBcal(T)\Hcal$, repeated successive measurements
in $\Acal$ and $\Bcal$ will give with probability one the outcomes $S$ and $T$. In other words, even if the commutator does not vanish, it is still possible for the measurement in $\Bcal$ to yield the outcome $T$, without perturbing the previous outcome $S$. However, if $[\PiAcal(S), \PiBcal(T)]\not=0$, there also exist pre-measurement states $\psi$ for which this is not true. To see this, note that, since the projectors don't commute, we know $\PiAcal(S)\Hcal$ is not invariant under $\PiBcal(T)$. So there  exists $\psi\in\PiAcal(S)\Hcal$ so that   $\PiBcal(T)\psi\not\in \PiAcal(S)\Hcal$. If initially the system is in such a  state $\psi$, the outcome of an $\Acal$ measurement is $S$. Since $\PiBcal(T)\psi\not\in\PiAcal(S)\Hcal$, we know $\PiBcal(T)\psi\not=0$, so that the probability that a measurement in $\Bcal$ yields $T$ does not vanish. When this occurs, the post-measurement state is, up to normalization, $\PiBcal(T)\psi\not\in \PiAcal(S)\Hcal$. Hence the probability that a subsequent measurement of $A$ yields $S$ is strictly less than one. 

In conclusion, when  two projectors $\PiAcal(S)$ and $\PiBcal(T)$ do not commute, it may still happen for some pre-measurement states that repeated successive measurements in $\Acal$ and $\Bcal$ systematically yield the same outcomes $S$ and $T$. However, this does no longer happen for every pre-measurement state that yields the outcomes $S$ and $T$ in the first two measurements. 

A simple but instructive example illustrating this situation is this:
\begin{eqnarray*}
A&=&a_1(|a_1\rangle\langle a_1|+|a_2\rangle\langle a_2|)+a_3|a_3\rangle\langle a_3|, \quad a_1<a_3,\\
B&=&b_1(|b_1\rangle\langle b_1|+|b_2\rangle\langle b_2|)+b_3|b_3\rangle\langle b_3|, \quad b_1<b_3,
\end{eqnarray*}
with
$$
|b_1\rangle=\frac1{\sqrt2}(|a_1\rangle+|a_2\rangle),\quad |b_2\rangle=\frac1{\sqrt2}\left(\frac1{\sqrt2}(|a_1\rangle-|a_2\rangle)+|a_3\rangle\right),\quad |b_3\rangle=\frac1{\sqrt2}\left(\frac1{\sqrt2}(|a_1\rangle-|a_2\rangle)-|a_3\rangle\right).
$$
Let $S=\{1,2\}, T=\{1,2\}$, so that $\PiAcal(S)$ and $\PiBcal(T)$ are the eigenprojectors of $A$, $B$ with eigenvalue $a_1, b_1$. They do not commute but 
$\PiAcal(S)\Hcal\cap\PiBcal(T)\Hcal=\C|b_1\rangle$. 
Now, let $|\psi\rangle=\frac1{\sqrt2}(|b_1\rangle+|a_3\rangle)$. Then a measurement of $A$ yields the outcome $a_1$ with probability $1/2$; the post-measurement state is then $|b_1\rangle$. A subsequent measurement of $B$ yields the outcome $b_1$ without altering the state and hence without altering the outcomes of any further measurements of $A$, which will always yield $a_1$. In this situation, the measurement of $B$ has not disturbed the one of $A$. If, on the other hand, the pre-measurement state is
$|\psi\rangle=\frac1{\sqrt2}(|a_2\rangle+|a_3\rangle)$, then a first measurement of $A$ still yields $a_1$ with probability $1/2$, but now the resulting state 
is $|a_2\rangle$ and
$$
|a_2\rangle =\frac1{\sqrt2}\left(|b_1\rangle-\frac1{\sqrt2}(|b_2\rangle +|b_3\rangle)\right).
$$
 A subsequent measurement of $B$ yields $b_1$ with probability $3/4$ and the resulting state is then 
 $$
 \sqrt{\frac23}(|b_1\rangle-\frac1{\sqrt2}|b_2\rangle).
 $$ 
 This is not an eigenvector of $A$ and a second measurement of $A$ only yields $a_1$ again with probability $5/6$ (and $a_3$ with probability $1/6$). In other words, this time the measurement of $B$ has perturbed the initial value obtained for $A$. 
This example thus illustrates  that indeed,  when two operators do not commute, successive measurements may or may not perturb outcomes previously obtained. 

It is common practice in quantum mechanics to identify the incompatibility of observables with their noncommutativity~\cite{CTDL15}.
In view of the previous discussion, this yields a rather weak notion of incompatible since it means that in some -- but not in all -- circumstances the measurement of the second observable may perturb the outcome of a previous measurement of the first observable. This is of course the result of the fact that incompatibility is defined here as the negation of compatibility, which is a very strong notion: it requires that a measurement of $B$ \emph{never} disturbs a previous measurement of $A$. 

In what follows, we will say two bases $\Acal$, $\Bcal$ are compatible if $[\PiAcal(S),\PiBcal(T)]=0$ for all $S,T\in\llbracket 1, d\rrbracket$. 
 And we will say two such bases are incompatible if this is not true; in other words, if \emph{there exist} $S,T\in\llbracket 1,d\rrbracket$ so that $[\PiAcal(S),\PiBcal(T)]\not=0$.
We will see below that, to connect incompatibility to uncertainty and nonclassicality, a stronger notion of incompatibility is needed.

\section{Uncertainties}\label{s:uncertainties}
It is well known that one way to understand the Heisenberg uncertainty relation is to view it as
a property of the Fourier transform saying, loosely speaking, that a function and its Fourier transform cannot both be sharply localized. When made precise, this statement can take many forms, thoroughly reviewed in~\cite{FoSi97} (See also~\cite{WiWi21}). The  Heisenberg uncertainty relation (in the form proven by Weyl~\cite{We50}) expresses the localization of the wave function in terms of the variances of its position and momentum distributions.  Another form was proven much more recently and states that a wave function and its Fourier transform cannot both have their support in sets of finite volume. In other words, they cannot both vanish outside a  region of finite volume. This formulation is very similar to the general support uncertainty principles that we will consider below (Section~\ref{s:suppuncrel}) in the finite dimensional setting. They have attracted attention more recently, in the context of signal analysis~\cite{DoSta89, GhoJa11}, in the theory of the Fourier analysis on finite groups~\cite{Tao05, WiWi21}, and in the study of nonclassicality in quantum systems with a finite dimensional Hilbert space~\cite{SDB21}, which is our focus here. 

In order to give perspective on the notion of support uncertainty, we first briefly review some aspects of the Heisenberg, Robertson, and entropic uncertainty relations. In Section~\ref{s:HeRouncrel} we recall the well-known link between noncommutativity and uncertainty as expressed in the Robertson uncertainty relation and show that noncommutativity alone does not lead to very strong information about the uncertainty inherent in a state. We briefly review entropic uncertainty relations in Section~\ref{s:entuncrel} and turn to the support uncertainty relation in Section~\ref{s:suppuncrel}.

\subsection{The Heisenberg and Robertson uncertainty relations}\label{s:HeRouncrel}
Rather than viewing the original Heisenberg uncertainty relation as a property of the Fourier transform, one can alternatively (and equivalently) see it as a bound on the uncertainties of an arbitrary state $|\psi\rangle$ with respect to two conjugate variables $Q$ and $P$ satisfying the canonical commutation relation $[Q,P]= i$. This alternative viewpoint leads to the generalization shown by Robertson~\cite{Rob29}, and which   concerns arbitrary noncommuting observables $A$ and $B$: 
\begin{equation}\label{eq:dispunc2}
\Delta A\Delta B\geq \frac12|\langle \psi|C|\psi\rangle|,\quad \textrm{where}\ C=-i[A,B].
\end{equation} 
It relates the precision with which we can know the value of two noncommuting  observables $A$ and $B$ when the system under consideration is  in a quantum state $|\psi\rangle$. More precisely, when the right hand side of this uncertainty relation does not vanish, the inequality can be interpreted as providing a limit on the statistical variation of the outcomes of measurements of $A$ and $B$ on an ensemble of systems in state $|\psi\rangle$~\cite{NiChu10, Coles17}. In other words, it quantifies to what extent the two observables are incompatible in this state: the sharper the knowledge on one, the less sharp the knowledge on the other. However, in finite dimensional Hilbert spaces, this can never give a useful lower bound on the product of the variances \emph{for all states $|\psi\rangle\in\Hcal$ at once}. Indeed, when $|\psi\rangle$ is an eigenstate of $A$ or of $B$, the right hand side vanishes and the inequality does not give a useful lower bound on the variances at all.  In other words, in a finite dimensional Hilbert space, one always has that
\begin{equation}\label{eq:Cmin0}
\min_{\langle\psi|\psi\rangle=1} |\langle \psi|C|\psi\rangle|=0.
\end{equation}
In fact, since $\Tr C=0$, the self-adjoint operator $C$ cannot have purely positive or purely negative spectrum: if it has a positive eigenvalue, it must also have a negative one, and vice-versa. As a result, the quadratic form $\langle \psi|C|\psi\rangle$ cannot be positive or negative definite, which again implies~\eqref{eq:Cmin0}. 
As a result, to get any information on the uncertainty on the knowledge of the values of the observables $A$ and $B$ in the state $|\psi\rangle$, one needs to know something about the state $|\psi\rangle$, namely the value of $\langle\psi|C|\psi\rangle$. When it vanishes, no information is obtained. In this sense, the noncommutativity of $A$ and $B$ alone does not provide automatically an information on the degree of  incompatibility of the two observables for any state $|\psi\rangle$. Again we see here that the common practice of equating the ``incompatibility'' of two quantum observables with the noncommutativity of the corresponding observables, recalled in the previous section, yields only limited information and suggests stronger notions of incompatibility are needed. We will extensively come back to this point below.

One may think the above observations are particular to the finite dimensional situation, but as we now briefly indicate, they are not. Let us therefore consider two observables represented by bounded self-adjoint operators $A$ and $B$ on an infinite dimensional Hilbert  space. They may not have any normalizable eigenvectors, and only continuous spectrum. The above arguments leading to~\eqref{eq:Cinf0} do then not apply  but  one still has the following simple result, closely related to Putnam's theorem~\cite{Cycon87}. The proof is given in Appendix~\ref{s:commutator}.
\begin{proposition} \label{prop:Cinf0} Let $A$ and $B$ be self-adjoint bounded operators on a finite or infinite dimensional Hilbert space and let $C=-i[A,B]$. Then
\begin{equation}\label{eq:Cinf0}
\inf_{\langle\psi|\psi\rangle=1} |\langle \psi|C|\psi\rangle|=0.
\end{equation}
\end{proposition}
Note that~\eqref{eq:Cinf0} again implies that the Heisenberg-Robertson uncertainty relation in Eq.~\eqref{eq:dispunc2} does not give a uniform bound on the variances for all $|\psi\rangle$ at once so that, here as well, noncommutativity alone does not imply an information on the degree of incompatibility of the observables $A$ and $B$. 

Note that, if $A=Q$ and $B=P$, then  $C=\bbone$ and of course~\eqref{eq:Cinf0} does not hold in that case. A slight variation on this example is $A=h(Q)$, $B=P$, so that $C=h'(Q)$. Assuming $h'(q)>\delta>0$ for all $q$, again~\eqref{eq:Cinf0} does not hold. As a concrete example, one may consider $h(q)=q+a\sin q$ with $a<1$. Then $h'(q)\geq (1-a)>0$ and $\inf_{\langle\psi|\psi\rangle=1} |\langle \psi|C|\psi\rangle|=1-a$.  Of course, in none of these cases are $A$ and $B$ bounded. 
These examples show that, in infinite dimension, the commutator $C$ can be a positive definite operator which has positive spectrum  bounded away from zero.  This cannot happen however, if both $A$ and $B$ are bounded, as a result of~\eqref{eq:Cinf0}. In that case, it is still possible for $C$ to be non-negative, meaning that its spectrum is non-negative, and includes $0$ so that  \eqref{eq:Cinf0} still holds. Examples of such operators are given in~\cite{How87, Ka91, HeKr19}.

 The conclusion we draw is therefore that, generally, the fact that two observables do not commute does not necessarily provide a useful bound
 on the uncertainties inherent in their measurement on an arbitrary state, as measured by the variances. As in  the discussion in the previous section, this implies again that the standard notion of incompatibility of two observables (namely their noncommutativity) does not imply much about such uncertainty. These observations provide a first motivation to look for other means than the variance to express the uncertainty inherent in a state with regards to the measurement of two observables. 
 
 Before doing so, let us point out that, in finite dimension, there is a second drawback to the use of variances to measure uncertainties. Indeed, when measurements are made in two bases $\Acal$ and $\Bcal$ that are not naturally associated with observables $A$ and $B$, then mean values and variances cannot be meaningfully defined~\cite{De83, MaUf88}. Different approaches to the quantification of uncertainty are then needed. The entropic and support uncertainty relations provide such alternatives, as discussed next.

\subsection{Entropic uncertainty relations}\label{s:entuncrel}
The two shortcomings of~\eqref{eq:dispunc2} in finite dimensional Hilbert space pointed out above are one of the reasons that entropic uncertainty relations have been developed in this context~\cite{De83, MaUf88, Coles17}. The best known of those is the Maassen-Uffink uncertainty relation which reads
\begin{equation}\label{eq:maassenuff}
H_{\Acal, \Bcal}(\psi)=H_\Acal(\psi)+H_\Bcal(\psi)\geq \ln \Mab^{-2},
\end{equation}
where  $\Acal, \Bcal$ are two orthonormal bases of $\Hcal$, as above, and 
$$
H_\Acal(\psi)=-\sum_i |\langle a_i|\psi\rangle|^2\ln|\langle a_i\psi\rangle|^2,\quad H_\Bcal(\psi)=-\sum_i |\langle a_i|\psi\rangle|^2\ln|\langle a_i\psi\rangle|^2
$$ 
are the Shannon entropies associated to the $\Acal$- and $\Bcal$-representations of $|\psi\rangle$.  Note that, when $\Acal,\Bcal$ are the eigenbases of two observables $A$ and $B$, the entropic uncertainty $H_{\Acal, \Bcal}(\psi)$ does not depend on the spectra of $A$ and $B$, only on the eigenbases, contrary to the variances $\Delta A, \Delta B$. In addition, $H_\Acal$ and $H_\Bcal$ are insensitive to permutations of the basis vectors and to global phase changes.

The entropy $H_\Acal(\psi)$ vanishes if and only if $|\psi\rangle$ equals one of the basis vectors; in that case the uncertainty inherent in $|\psi\rangle$ with respect to an $\Acal$-measurement vanishes.  An analogous statement holds for $H_\Bcal$. The entropy $H_{\Acal, \Bcal}(\psi)$ yields a measure of the joint uncertainty inherent in $|\psi\rangle$  with respect to $\Acal$ and $\Bcal$ measurements. It vanishes iff $|\psi\rangle$ is a common basis vector of $\Acal$ and $\Bcal$. In fact,  such a vector exists if and only if  $\Mab=1$. 
In that case, the Maassen-Uffink uncertainty relation provides no information, since the Shannon entropies are always nonnegative. When $\Mab<1$, on the other hand, it provides a uniform lower bound valid for all $|\psi\rangle$, similar to the Heisenberg uncertainty relation for conjugate observables. We will see below that the condition $\Mab<1$ can be interpreted as a -- very weak -- incompatibility condition between the bases: see Proposition~\ref{prop:incompatibilities}.
 The smaller $\Mab$, the stronger the bound obtained. From~\eqref{eq:mMbounds} one concludes therefore that the strongest entropic uncertainty relation is  obtained when $\Mab=\frac1{\sqrt d}$, in which case all overlaps have the same amplitude and the bases are mutually unbiased. The lower bound is then reached (at least) by the basis vectors.  
We note however that, even when $\Mab<1$, the bound is not always optimal and may not be reached, as extensively discussed in~\cite{AbEtAl15}.


\subsection{The support uncertainty relation} \label{s:suppuncrel}
Another uncertainty relation that has been used in various contexts other than quantum mechanics is the support uncertainty relation (see Eq.~\eqref{eq:uncprincab} below), first introduced in~\cite{DoSta89} in the context of the Fourier transform on finite groups. It has been shown to have much wider validity, however: see~\cite{GhoJa11, WiWi21}. 

To describe it, we introduce first the $\Acal$-support and $\Bcal$-support of $|\psi\rangle$:
$$
S_\psi=\{i\in\llbracket 1, d\rrbracket \mid \langle a_i|\psi\rangle\not=0\},\quad
\mathrm{and} \quad
T_\psi=\{j\in\llbracket 1, d\rrbracket \mid \langle b_j|\psi\rangle\not=0\}.
$$
They are the supports of the state $|\psi\rangle$ in the $\Acal$- and $\Bcal$-representations introduced in Eq.~\eqref{eq:ABreps}. One can also think of $S_\psi$ as the set of possible outcomes of a fine-grained measurement in the $\Acal$-basis with the pre-measurement state $|\psi\rangle$. We then define
\begin{equation}\label{eq:supportsizes}
n_\Acal(\psi)=|S_\psi|\not=0,\ n_\Bcal(\psi)=|T_\psi|\not=0,
\end{equation}
where  $|S|$ denotes the number of elements in $S$. 
Then $n_\Acal(\psi)$ (respectively $n_\Bcal(\psi)$) is the number of nonvanishing entries of $u(\psi)$ (respectively $v(\psi)$); $n_\Acal(\psi)$ is therefore  the size of the support or the ``spread'' of the probability distribution $|u_i(\psi)|^2$ (respectively $|v_j(\psi)|^2$) of the state $|\psi\rangle$ in the $\Acal$- (respectively $\Bcal$)-representations. As such $\na(\psi)$ and $\nb(\psi)$ are measures of the localization of $|\psi\rangle$ in the respective representations. They can be thought of as providing a measure of the uncertainty on the outcomes of a measurement in the $\Acal$-basis, respectively $\Bcal$-basis when the system is in the state $|\psi\rangle$. 

In our notation the support uncertainty relation then reads
\begin{equation}\label{eq:uncprincab}
n_\Acal(\psi)n_\Bcal(\psi)\geq \Mab^{-2}.
\end{equation}
 The analogy with the Heisenberg-Robinson uncertainty relation is clear: it says that the supports of $|\psi\rangle$ in the $\Acal$- and in the $\Bcal$-representation cannot both be small since their product is bounded below.  Like the Robertson uncertainty principle, it has a very simple proof, that we provide for completeness in Appendix~\ref{s:proof_uncprinc}. (See also~\cite{SDB21} for an alternatiive argument.)  It can also be seen as an immediate consequence of the more sophisticated 
Maassen-Uffink bound since the Shannon entropy is well-known to have the property that
$$
H_\Acal(\psi)\leq \ln(\na(\psi)),\quad H_\Bcal(\psi)\leq \ln(\nb(\psi)).
$$
 Unlike what happens in the Heisenberg-Robinson uncertainty relation, and similarly to what happens for the Maassen-Uffink entropic inequality, the lower bound in Eq.~\eqref{eq:uncprincab} has the advantage that it does not depend on the state $|\psi\rangle$,  and only depends on $\Acal$ and $\Bcal$ through the maximum value  $\Mab$ of the matrix elements of $U$. It is also not necessarily saturated, as we will see below. Like the Maassen-Uffink bound, it is without interest if $\Mab=1$, since both $\na(\psi)$ and $\nb(\psi)$ are positive integers. 
 
For the purpose of relating uncertainty to incompatibility and KD nonclassicality, it turns out to be more instructive to consider lower bounds on the sum 
\begin{equation}\label{eq:suppunc}
n_{\Acal,\Bcal}(\psi):=n_{\Acal}(\psi)+ n_{\Bcal}(\psi),
\end{equation}
 that we shall refer to as the \emph{support uncertainty} of $|\psi\rangle$, rather than on the product $\na(\psi)\nb(\psi)$. Both such bounds yield information on the tradeoff between uncertainty in the $\Acal$- and $\Bcal$-representations of the state, but, as we shall establish, the sum provides a more natural and direct link with a strong form of incompatibility and KD nonclassicality.   
The situation is not unlike the one in quantum optics, where the sum of the variances $(\Delta Q)^2$ and $ (\Delta P)^2$ of two conjugate quadratures $Q,P$ of the field yields a measure of the uncertainty and of the optical nonclassicality of the state referred to as its \emph{total noise}~\cite{hi89}. It satisfies the lower bound $(\Delta Q)^2+(\Delta P)^2\geq1$ for all states.

For later purposes, we also introduce the \emph{minimal support uncertainty} $\nabmin$ of the two bases $\Acal, \Bcal$ as
\begin{equation}\label{eq:nabmin}
\nabmin=\min_{|\psi\rangle\not=0} n_{\Acal,\Bcal}(\psi).
\end{equation}
Clearly
$$
2\leq \nabmin\leq d+1.
$$
The lower bound is obvious since for any nonzero $|\psi\rangle\in\Hcal$, $\na(\psi)\geq 1 $ and $\nb(\psi)\geq 1$. For the upper bound, consider any basis vector $|a_i\rangle$. Then $\na(|a_i\rangle)=1$, $\nb(|a_i\rangle)\leq d$. So $\nab(|a_i\rangle)\leq d+1$, which yields the upper bound. Note that 
\begin{equation}\label{eq:minab3}
\nabmin\geq 3 \Leftrightarrow \Mab<1.
\end{equation}
Indeed, $\nabmin\geq 3$ is equivalent to the fact that each $\Acal$ basis vector is the linear combination of at least two $\Bcal$-basis vectors and vice versa. Hence $\Mab<1$. The converse is obvious. More generally, Eq.~\eqref{eq:uncprincab} implies that
\begin{equation}\label{eq:uncprincaplusb}  
\nabmin\geq \frac{2}{\Mab},
\end{equation}
which is again not sharp in general, as can be observed in Fig.~\ref{fig:uncdiagTAO}. See also the next subsection.

We have already pointed out that the Maassen-Uffink and support uncertainty relations provide no information when $\Mab=1$. This is not surprising. 
Indeed, when $\Mab=1$, $\Acal$ has $k\geq1$ basis vectors $|a_i\rangle$ that are equal (up to an irrelevant phase) to $k$ basis vectors $|b_j\rangle$. It is then not possible to get a nontrivial bound constraining the uncertainties $\na(\psi), \nb(\psi)$ since they can both be simultaneously equal to $1$. As anticipated at the end of Section~\ref{s:setup}, after reordering the basis vectors and adjusting their phase, we can then assume $|a_1\rangle=|b_1\rangle, \dots, |a_k\rangle=|b_k\rangle$ and split the Hilbert space $\Hcal$ into the subspace $\Hcal_k$ generated by the first $k$ basis vectors and its orthogonal complement $\Hcal_k^\perp$, of dimension $d'=d-k$. 
The analysis of the uncertainty inherent in a state $|\psi\rangle$ can then be done separately for its component in each of these subspaces.


\subsection{Uncertainty diagrams}
Given two bases $\Acal$ and $\Bcal$, their \emph{uncertainty diagram} UNCD$(\Acal, \Bcal)$ was defined in~\cite{SDB21} to be the subset of all points $(\na,\nb)$ in the first quadrant of the $\na\nb$-plane for which there exists a state $|\psi\rangle\in\Hcal$ so that 
\begin{equation}\label{eq:uncdiagdef}
\na=\na(\psi),\quad \nb=\nb(\psi),
\end{equation}
In other words, 
\begin{equation}\label{eq:uncdiagdef2}
\UNCD(\Acal, \Bcal)=\{(\na,\nb)\in \mathbb N^2 | \exists \psi\in \Hcal, \na(\psi)=\na, \nb(\psi)=\nb\}.
\end{equation}
One can also view the uncertainty diagram is  the image of the map $|\psi\rangle\to(\na(\psi),\nb(\psi))\in\llbracket 1,d\rrbracket^2$.  As we will see, the uncertainty diagram is a useful practical and visual tool to study the mutual uncertainty of a state with respect to the $\Acal$- and $\Bcal$- representations. Analogous uncertainty diagrams for entropic uncertainty measures were introduced in~\cite{AbEtAl15}.
 In view of Eq.~\eqref{eq:uncprincab}, the uncertainty diagram always lies above the hyperbola
$$
\na\nb=\Mab^{-2}.
$$
This lower bound is not necessarily reached, of course, since in general $\Mab^{-2}$ is not an integer.   

Three examples of uncertainty diagrams are presented in Fig.~\ref{fig:uncdiagTAO}. In the left panel a spin $1$ system is considered, for which $\Hcal=\C^3$. We choose for $\Acal$ and $\Bcal$ the eigenbases of $J_z$ and $J_x$, where
$$
J_z=\begin{pmatrix}
1&0&0\\
0&0&0\\
0&0&-1
\end{pmatrix},\quad
J_x=\frac1{\sqrt2}\begin{pmatrix}
0&1&0\\
1&0&1\\
0&1&0
\end{pmatrix}.
$$
We write $|z,\epsilon\rangle$ and $|x,\epsilon\rangle$, $\epsilon=1,0,-1$ for the corresponding eigenvectors. The transition matrix from the $J_z$ basis to the $J_x$ basis is then
$$
U=
\frac12
\begin{pmatrix}
1&\sqrt 2&1\\
\sqrt 2&0&-\sqrt 2\\
1&-\sqrt 2&1
\end{pmatrix}.
$$
In this case, $\mab=0$, $\Mab=1/\sqrt2$ and the support uncertainty relation reads $n_z n_x(\psi)\geq 2$. The corresponding lower bound on the uncertainty diagram is therefore $n_z n_x=2$.   Note that $J_x$ and $J_z$ do not commute so that the bases are incompatible.   Details of the computations underlying the left panel of Fig.~\ref{fig:uncdiagTAO} are given in Appendix~\ref{app:spinone}.

 \begin{figure*}[t!]
\begin{center}
\includegraphics[height=3cm]{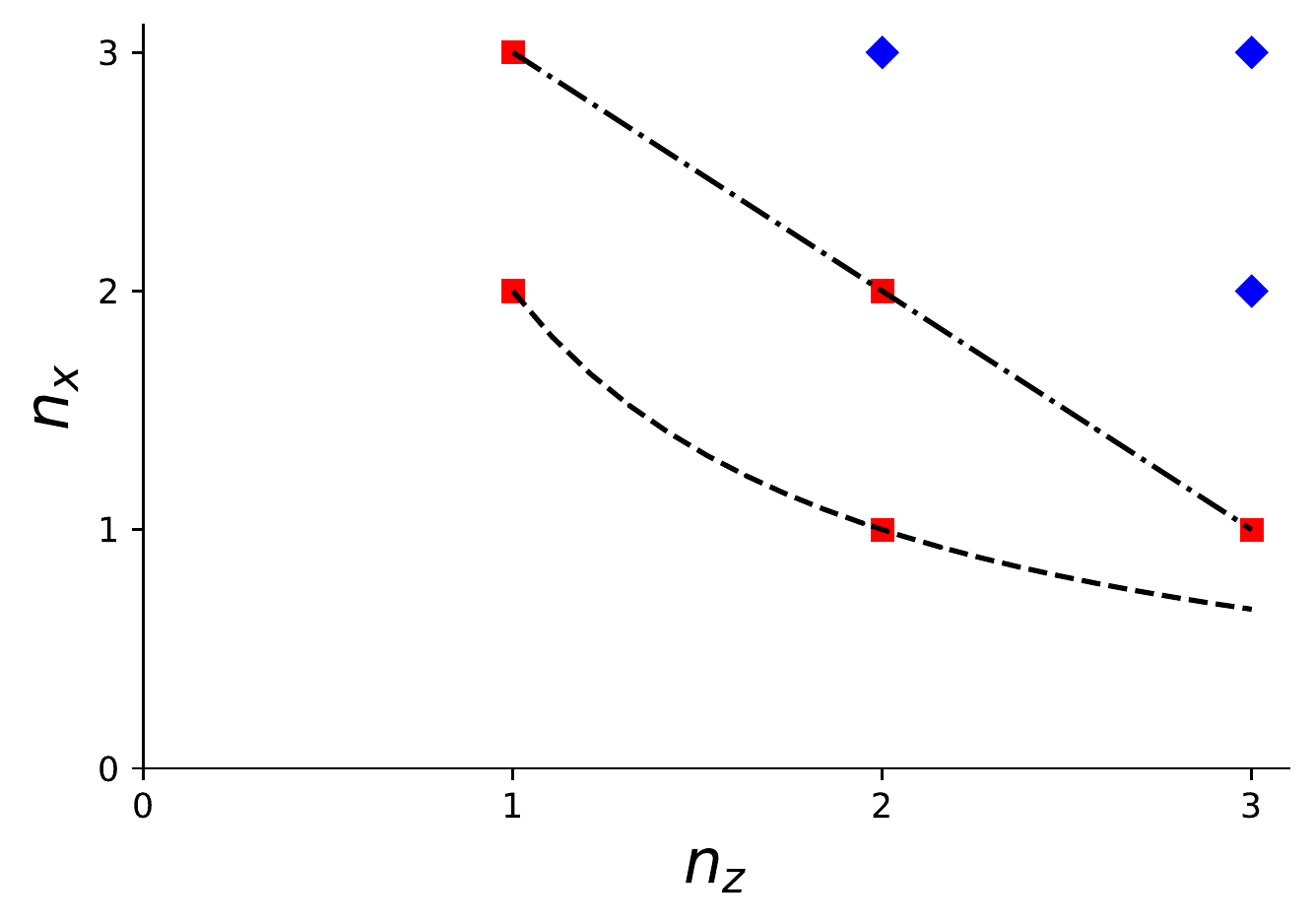}\hskip0.5cm
\includegraphics[height=3cm, keepaspectratio]{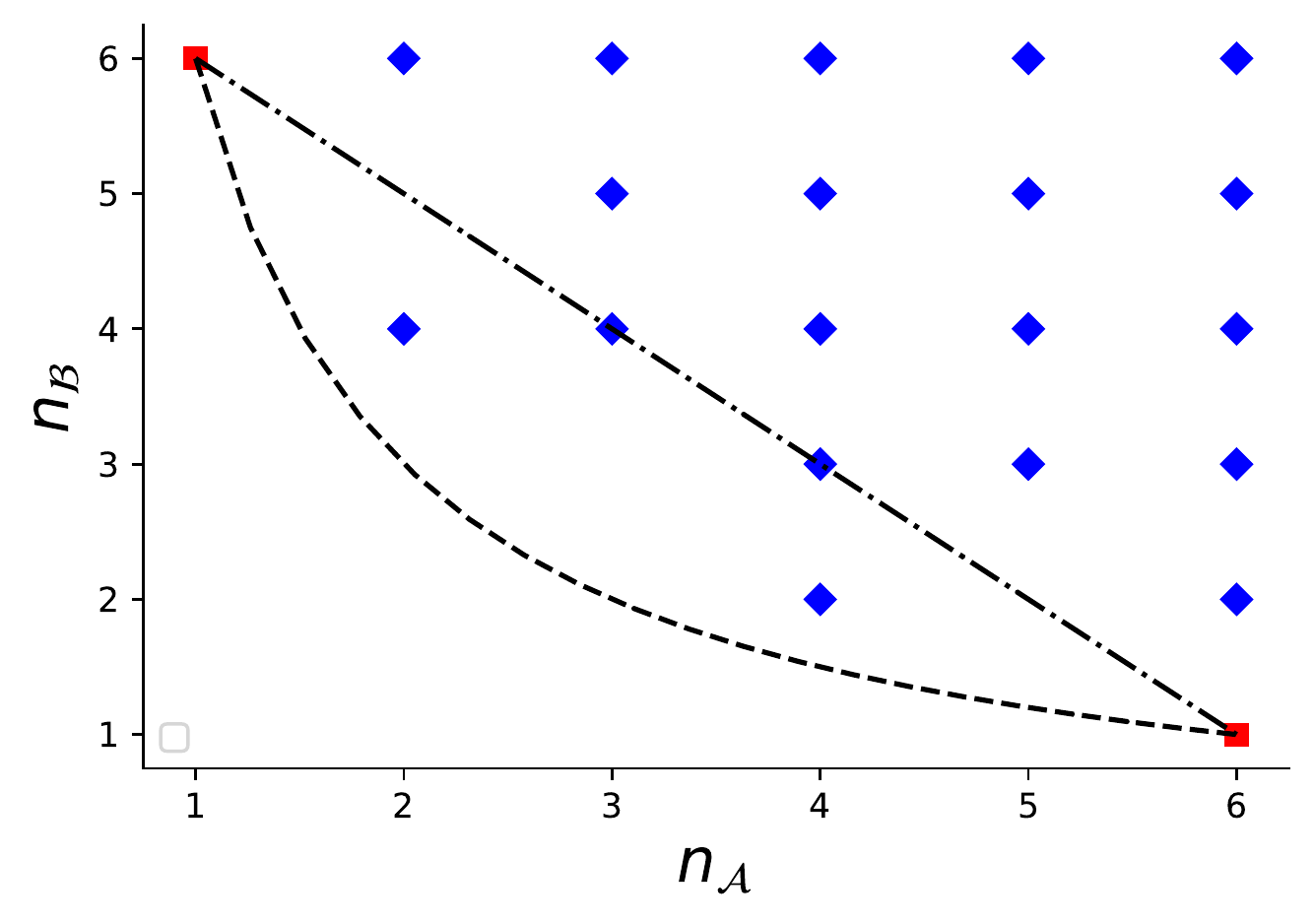}\hskip0.5cm
\includegraphics[height=3cm, keepaspectratio]{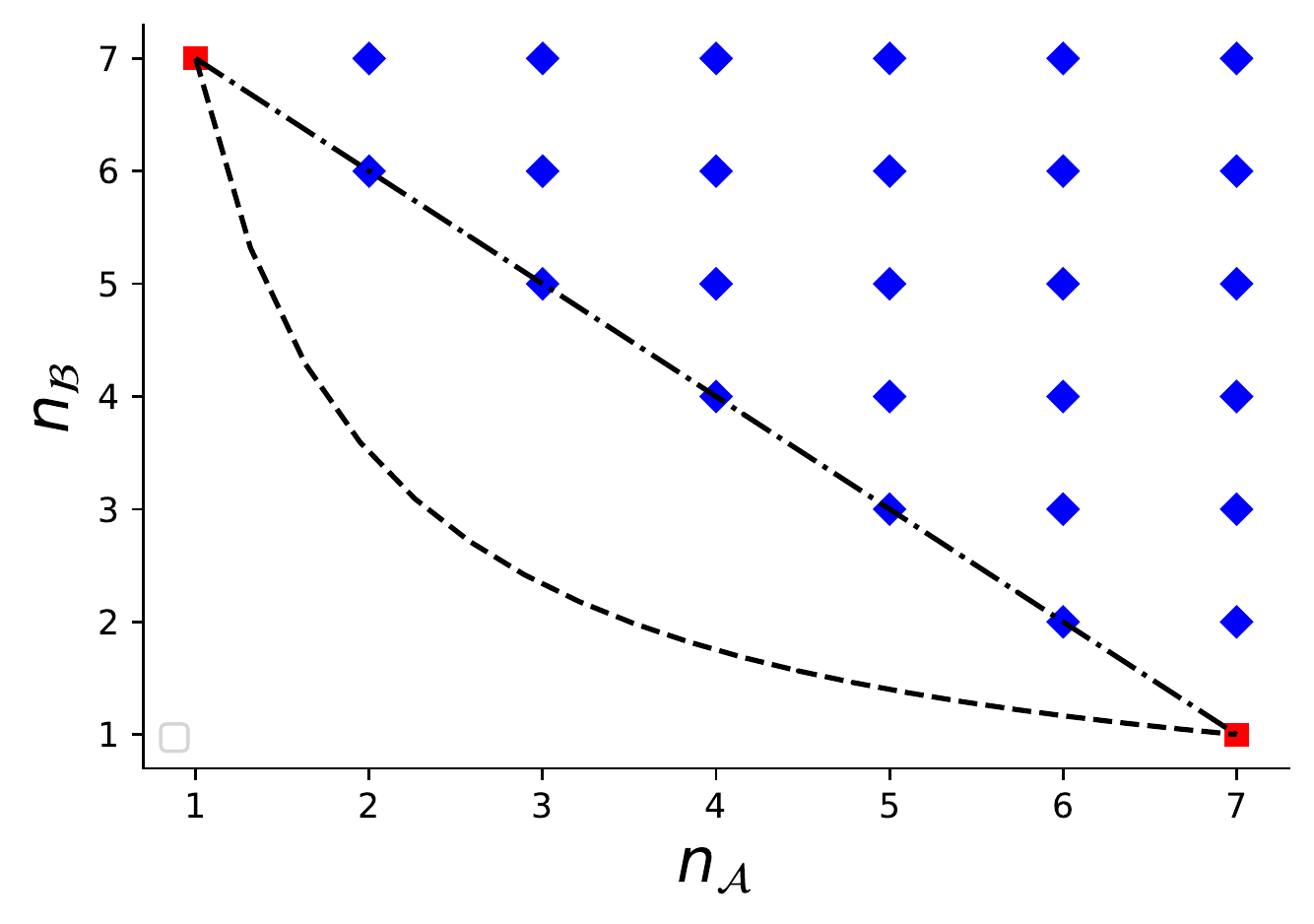}
\end{center}
\caption{Uncertainty diagrams. Dot-dashed line: $n_\Acal(\psi)+n_\Bcal(\psi)=d+1$. Dashed curve: $n_\Acal(\psi)n_\Bcal(\psi)=\Mab^{-2}$. Diamonds (blue): KD-nonclassical states. Squares (red): KD-classical states. Left panel: Case of the spin $1$ transition matrix ($d=3$) for which $\mab=0, \Mab=1/\sqrt{2}$.  Middle panel. Case of the Tao matrix $(d=6)$ which is a MUB for which therefore $\mab=\frac1{\sqrt d}=\Mab$.     Note the absence of states for which $(\na,\nb)=(2,5), (3,3)$ or $(5,2)$.   Right panel. Case of the DFT matrix  for $d=7$ which is both  a MUB and COINC. 
  \label{fig:uncdiagTAO}}
\end{figure*}

The middle panel of Fig.~\ref{fig:uncdiagTAO} shows the uncertainty diagram for two basis $\Acal,\Bcal$ for which the transition matrix is given by the Tao matrix~\cite{Tao04}, which is the transition matrix for two MUB in dimension $d=6$ given by
\begin{equation}\label{eq:taomatrix}
U=\frac1{\sqrt d}\\
\begin{pmatrix}
1&1&1&1&1&1\\
1&1&\omega&\omega&\omega^2&\omega^2\\
1&\omega&1&\omega^2&\omega^2&\omega\\
1&\omega&\omega^2&1&\omega&\omega^2\\
1&\omega^2&\omega^2&\omega&1&\omega\\
1&\omega^2&\omega&\omega^2&\omega&1
\end{pmatrix},
\quad 
\omega=\exp(i\frac{2\pi}{3}).
\end{equation}
The details of the computations underlying the middle panel of Fig.~\ref{fig:uncdiagTAO} are given in Appendix~\ref{app:tao}. 

The right hand panel of Fig.~\ref{fig:uncdiagTAO} shows the uncertainty diagram of two bases related by the discrete Fourier transform (DFT) in dimension $d=7$. Its structure is fully explained by Theorem~\ref{thm:COINCsupport_new} below, as we shall see. 

In general, given two bases $\Acal$ and $\Bcal$, it is not clear what the precise shape of its uncertainty diagram is. For example, it is not clear what its lower edge $n_\Acal\in\llbracket 1, d\rrbracket \to L(\na)$ is, defined as
\begin{equation}\label{eq:loweredge}
L(\na)=\min\{\nb(\psi)| \psi\in\Hcal, \na(\psi)=\na, \|\psi\|=1\}.
\end{equation}
We know from the support uncertainty relation that $L(\na)\geq \Mab^{-2}\na^{-1}$, but this bound may not be reached for all values of $\na$.  Also, there may be ``holes'' in the uncertainty diagram, meaning values $\nb>L(\na)$ for which there exists no $|\psi\rangle$ so that $\na(\psi)=\na, \nb(\psi)=\nb$.  Both these phenomena are illustrated in Fig.~\ref{fig:uncdiagTAO}. For example, for the Tao matrix, one has
$$
L(1)=6,\ L(2)=4=L(3),\ L(4)=2,\ L(5)=3,\ L(6)=1.
$$

Further examples of uncertainty diagrams are given in~\cite{SDB21}. 
Similar difficulties arise with the uncertainty diagrams for entropic uncertainty, as studied in~\cite{AbEtAl15}. 

It turns out that for completely incompatible bases, to which we turn our attention next (Definition~\ref{def:COINC}), the lower edge of the uncertainty diagram is easily determined to be
$$
L(\na)=d+1-\na.
$$
In addition, it turns out the uncertainty diagram then has no ``holes''. This is shown in Theorem~\ref{thm:COINCsupport_new} below.


\section{Complete incompatibility} \label{s:COINC}
\subsection{Complete incompatibility: definition, interpretation.}\label{s:coinc_def}
The previous sections make it clear that incompatiblity, understood as noncommutativity, is a rather weak notion. On the one hand, it only implies that a  measurement in $\Bcal$ \emph{may} perturb a previously obtained outcome in a measurement in $\Acal$. On the other hand, it provides only limited information on the (minimal) uncertainty inherent in an arbitrary given state. To avoid these shortcomings, the following very stringent notion of incompatibility was introduced in~\cite{SDB21}.
\begin{definition}\label{def:COINC}
We say that two bases $\Acal$ and $\Bcal$ are completely  incompatible (COINC) if and only if all index sets $S,T$ in $\llbracket 1,d\rrbracket$ for which $|S|+|T|\leq d$ have the property that $\PiAcal(S)\Hcal\cap\PiBcal(T)\Hcal=\{0\}$. 
\end{definition}
It follows from the discussion in Section~\ref{s:compatibility}, and particularly the one concerning Eq.~\eqref{eq:nonperturb}, that this mathematical property has a direct interpretation in terms of the incompatibility of repeated successive measurements in the $\Acal$ and $\Bcal$ bases. In short, complete incompatibility of $\Acal$ and $\Bcal$ means the following. Suppose that the system is prepared in \emph{any} state $|\psi\rangle$ and a measurement in the $\Acal$ basis yielding the outcome $S$ is followed by one in the $\Bcal$ basis yielding the outcome $T$. In that case, this second measurement has \emph{always} a nonvanishing probability to perturb the result of the first, provided $|S|+|T|\leq d$.  Since Eq.~\eqref{eq:nonperturb} allows for the possibility of the second measurement not to perturb the first for certain pre-measurement states, it is by negating this condition that a very strong notion of incompatibility is obtained. This is all the more true since the condition is imposed for all outcomes $S,T$, with $|S|+|T|\leq d$.

Note that the restriction $|S|+|T|\leq d$ is needed since, when $|S|+|T|>d$, the intersection $\PiAcal(S)\Hcal\cap\PiBcal(T)\Hcal$ is necessarily nontrivial because the dimension of $\PiAcal(S)\Hcal$ is $|S|$ and that of $\PiBcal(T)\Hcal$ is $|T|$. What this expresses is the intuitively expected fact that if the measurements are too coarsegrained, then it may happen that one does not perturb the other. Note also that it would be sufficient in the definition to require $\PiAcal(S)\Hcal\cap\PiBcal(T)\Hcal=\{0\}$ for all $S,T$ so that $|S|+|T|=d$, since the more general requirement follows from this immediately. 

The condition $\Pi_\Acal(S)\Hcal\cap \Pi_\Bcal(T)\Hcal=\{0\}$ can also be understood geometrically, a viewpoint that sometimes helps the intuition.  Given $S$, respectively $T$, one can think of $\PiAcal(S)\Hcal$, respectively $\PiBcal(T)\Hcal$ as  ``coordinate vector spaces'' for the $\Acal$-, respectively $\Bcal$-representations. Indeed, analogously with the ``$xy$-plane'', which corresponds to points in $\R^3$ that have vanishing $z$-coordinate, these spaces contain all states  that have vanishing coordinates corresponding to indices in $S^c$, respectively $T^c$. The condition that $\Pi_\Acal(S)\Hcal\cap \Pi_\Bcal(T)\Hcal=\{0\}$ for all $|S|+|T|\leq d$ can then be understood as saying that these coordinate spaces for the $\Acal$- and $\Bcal$-representation have a trivial intersection. To put it differently, the coordinate spaces for $\Acal$ sit askew in the coordinate spaces for $\Bcal$, and vice versa. 

The definition of COINC bases depends on the two bases $\Acal$ and $\Bcal$ only through the unitary transition matrix between them, defined as  $U|a_j\rangle=|b_j\rangle$, with matrix elements $U_{ij}=\langle a_i|b_j\rangle$ in the $\Acal$-basis. This can be seen as follows. If $\Acal, \Bcal$ have transition matrix $U$ and $\Acal', \Bcal'$ have transition matrix $U'$, then $U=U'$ if and only if there exists a unitary map $V$ so that $|a'_i\rangle =V|a_i\rangle$ and $|b'_i\rangle=V|b_i\rangle$. But in that case, clearly, $\Acal$ and $\Bcal$ are COINC if and only if $\Acal'$ and $\Bcal'$ are. 
Conversely, given an arbitrary unitary $d\times d$ matrix $U$, one can always view it as the transition matrix between the canonical basis of $\Hcal=\C^d$ (taken to be $\Acal$) and the basis composed of the columns of $U$ (taken to be $\Bcal$).  In view of this, we will say a unitary matrix $U$ is COINC whenever these two bases $\Acal$ and $\Bcal$ are. Depending on the situation at hand, it is sometimes more convenient to think in terms of two bases $\Acal, \Bcal$ and sometimes in terms of a unitary matrix $U$, without reference to the bases. Note finally that the complete incompatibility of two bases is also not affected by a renumbering, nor by global phase changes of the basis vectors.

\subsection{Complete incompatibility: characterization and examples}
 A number of examples of COINC bases were given in~\cite{SDB21}. 
It is first of all easy to see that in dimension $d=2,3$, two bases are COINC if and only if $\mab>0$, in other words when the transition matrix has no zeros. This is no longer true in higher dimension as we shall show in Proposition~\ref{prop:incompatibilities} below. For example,  it is shown in~\cite{SDB21} using result of Tao~\cite{Tao04} that, when the transition matrix $U$ is the discrete Fourier transform (DFT), the bases are COINC if and only if $d$ is prime. 
Other explicit examples of bases that are COINC don't seem to be very easy to come by.  It was nevertheless illustrated numerically in~\cite{SDB21} that small perturbations of MUB for $d=4$ tend to be COINC. Theorem~\ref{thm:COINC_U_dense} below explains this numerical observation in much more generality.
\begin{theorem}\label{thm:COINC_U_dense}
The set of completely incompatible unitary $d\times d$ matrices is open and dense in the set of all unitary matrices.
\end{theorem}
The proof, given in  Appendix~\ref{app:coincopendense}, relies on the following algebraic characterization of complete incompatibility proven in~\cite{SDB21}. 
\begin{lemma}\label{lem:minorcrit}
Two bases $\Acal$ and $\Bcal$ are COINC if and only if none of the minors of their transition matrix $U$  vanishes.
\end{lemma}
To make this paper self-contained we also provide a proof of  the lemma in  Appendix~\ref{app:coincopendense}.

The theorem may seem somewhat surprising at first sight, since we insisted on the fact that complete incompatibility is a strong notion of incompatibility. To understand why it is true, one may consider Lemma~\ref{lem:minorcrit}. A minor is a multi-variable polynomial in the matrix elements of $U$ and its zeros form a very ``thin'' set: essentially it is a hypersurface of co-dimension $1$. Hence, since small perturbations of $U$ will typically change the values of all its minors, those that initially vanish will typically become nonvanishing. On the other hand, if the perturbation is small enough, those that do not vanish initially will remain nonvanishing.  This suggest that the set of COINC transition matrices is dense. This is precisely what we prove in Appendix~\ref{app:coincopendense}.  The proof is somewhat technical because of the need to ensure that the perturbations considered preserve the unitarity. By the same reasoning, the set of COINC transition matrices is dense.  Indeed, if no minor of $U$ vanishes, then this will remain true for all small perturbations of $U$.  The situation is similar to what happens when one considers commuting self-adjoint $d\times d$ matrices, which are considered compatible.  Those also form a very thin subset of all pairs of self-adjoint matrices and again, a slight perturbation will typically render them noncommuting, hence incompatible. And in this context also, when two observables do not commute, small perturbations of them will still not commute. This being said, the set of completely incompatible observables is considerably smaller than the set of merely incompatible ones, as will become clear in the following section.


\subsection{A hierarchy of incompatibilities}
As a first illustration of the strength of the notion of complete incompatibility, let us consider the case where  $A$ and $B$ are observables with non-degenerate spectrum and eigenbases $\Acal$ and $\Bcal$. Saying that $A$ and $B$ are incompatible in the standard sense that they do not commute then means that there exist at least one $i,j$ so that $[|a_i\rangle\langle a_i|, \langle b_j\rangle\langle b_j|]\not=0$. This is clearly much weaker than requiring that this needs to be  true for all $i,j$ which in turn is weaker than requiring that $[\PiAcal(S),\PiBcal(T)]\not=0$ for all  $S,T\subset \llbracket 1,d\rrbracket$ with $|S|+|T|\leq d$. 
This is readily illustrated on a spin $1$ system,  so that $\Hcal=\C^3$, taking $\Acal$ to be the eigenbasis of $J_z$ and $\Bcal$ of $J_x$. 
In spite of the fact that $J_x$ and $J_z$ do not commute, 
the corresponding eigenbases are not COINC. To see this, one may note that $|x,0\rangle\in\PiAcal(S)\Hcal$, with $S=\{-1,1\}$. So, if first a coarse-grained measurement is made of $J_z$ on an arbitrary state $|\psi\rangle$, with two possible outcomes ``$J_z=0$'' and ``$J_z\not=0$'', and the outcome is ``$J_z\not=0$'', then a subsequent fine-grained measurement of $J_x$ yielding the outcome ``$J_x=0$'' will not perturb the first measurement. 
This is reflected in the fact that  the outcomes $T=\{0\}$ and $S=\{-1, 1\}$ are compatible in the sense of Definition~\ref{def:compatibleoutcomes}: $[\PiAcal(S),\PiBcal(T)]=0$.

To shed further light on the definition of complete incompatibility and to justify the use of this terminology, we will now show that Definition~\ref{def:COINC} implies  three properties that can be seen as weaker manifestations of incompatibility. 
\begin{proposition}\label{prop:incompatibilities}  Let $\Acal, \Bcal$ be two orthonormal bases on a $d\geq 2$ dimensional Hilbert space. Consider the following statements: 
\begin{enumerate}
\item[(i)] $\Acal$ and $\Bcal$ are COINC.
\item[(ii)] For all $S,T\subset \llbracket 1,d\rrbracket$, with $1\leq |S|, |T|<d$, $[\PiAcal(S),\PiBcal(T)]\not=0$.
\item[(iii)] $\mab>0$.
\item[(iv)] $\Mab<1$.
\end{enumerate}
Then (i) $\Rightarrow$ (ii) $\Rightarrow$ (iii) $\Rightarrow$ (iv).\\ 
In addition
\begin{itemize}
\item When $d=2$, one has (i) $\Leftrightarrow$ (ii) $\Leftrightarrow$ (iii) $\Leftrightarrow$ (iv);
\item When $d=3$, one has (i) $\Leftrightarrow$ (ii) $\Leftrightarrow$ (iii) $\Rightarrow$ (iv), but (iv) does not imply (iii);
\item When $d\geq 4$, one has  (i) $\Rightarrow$ (ii) $\Rightarrow$ (iii) $\Rightarrow$ (iv), but (iv) does not imply (iii) and (iii)  does not imply (ii);
\item When $d=4$ or $d=6$, (ii) does not imply (i). 
\end{itemize}
\end{proposition}
The proof of the proposition is given below, but we first discuss its meaning by interpreting each of the four statements (i)-(iv) as a manifestation of incompatibility. 
The interpretation of (i) was given above. We need to understand the meaning of (ii), which is a strong statement of ``noncommutativity''.  Indeed, when two operators $A$ and $B$ commute, then all their spectral projectors commute. Saying that they don't commute means therefore that there exist some spectral projectors that don't commute; (ii) is much stronger, since it requires no spectral projectors to commute. 
For example, the spin 1 bases discussed above clearly do not satisfy (ii), as explained above. 

 To see that (ii) implies also a strong notion of incompatibility, first  recall that we showed  two outcomes $S$ and $T$ are compatible (Definition~\ref{def:compatibleoutcomes})  if and only if $[\PiAcal(S),\PiBcal(T)]=0$ (see Eq.~\eqref{eq:compatibleST}). So (ii) can be paraphrased by saying that no two outcomes $S$ and $T$ are compatible for a measurement in $\Acal$ followed by one in $\Bcal$. This is indeed a strong notion of incompatibility. It can in fact be viewed as the negation of a weak form of compatibility given  by the negation of (ii): there exist nontrivial $S, T$ so that   $[\PiAcal(S),\PiBcal(T)]=0$. However, note that  (ii) only means that for any two such outcomes, it \emph{may} happen -- depending on the pre-measurement state -- that the occurrence of the second one perturbs the first one. In other words, the second statement guarantees the possible perturbation of the first measurement by the second for all outcomes, but not for all states. One then understands why the second assertion is a weaker form of incompatibility than the first. Indeed, when (i) holds, the second outcome \emph{always} perturbs the first one, as explained above. In Appendix~\ref{s:counterexamples} we give  examples of bases $\Acal$ and $\Bcal$ satisfying (ii) but not (i) in dimensions $4$ and $6$. 
We conjecture it is true that (ii) does not imply (i) in all dimensions $d\geq 4$, but we have not produced such examples in other dimensions than $4$ and $6$. Note  that checking condition (ii) in higher dimension is necessarily more involved since there are then many more choices of $S$ and $T$.

Statement (iv), namely $\Mab<1$ guarantees that  none of the basis vectors of $\Acal$ is  a basis vector of $\Bcal$, and vice versa. This means that, when a first fine-grained measurement in $\Acal$ leaves the system in a basis state of $\Acal$, a subsequent measurement of $\Bcal$ has at least two possible outcomes with nonzero probability and therefore necessarily disturbs the first measurement, which is indeed a manifestation of incompatibility. So a fine-grained measurement in $\Bcal$ then necessarily perturbs a previous fine-grained measurement in $\Acal$. In particular, if $\Acal$ and $\Bcal$ are  eigenbases of two observables $A$ and $B$ with nondegenerate spectra, then (iv) implies the noncommutativity of all $|a_i\rangle\langle a_i|$ and $|b_j\rangle\langle b_j|$ which implies the noncommutativity of $A$ and $B$~\cite{footnote1}. The converse is not true: $A$ and $B$ may not commute and nevertheless have some eigenvectors in common, in which case $\Mab=1$.

Statement (iii) is a strengthening of (iv). It is sometimes referred to as a condition of ``complementarity'' for the two bases. This is because the condition $\mab>0$ implies that, when the system is in one of the basis states of $\Acal$ and a fine-grained measurement is made in $\Bcal$, then any outcome $|b_j\rangle$ may occur with nonvanishing probability. And vice versa. In that sense, there is uncertainty in the outcome of this second measurement, and this uncertainty is larger when $\mab$ is larger. Some authors~\cite{DuEnBeZy10} reserve the term ``complementary'' or ``maximally incompatible'' for mutually unbiased bases, for which these probabilities are the same for all $j$ so that $\mab$ takes on its maximal possible value $1/\sqrt d$.  When (iii) holds, repeated successive fine-grained measurements of $\Acal$ and $\Bcal$ do not  yield the same result with probability one.   In fact, they can give any result and in that sense there is incompatibility between $\Acal$ and $\Bcal$. The proposition asserts this form of incompatibility  does not imply the stronger condition on the commutators in (ii), in any dimension $d\geq 4$. 

We now turn to the proof of Proposition~\ref{prop:incompatibilities} which contains instructive examples and counterexamples.\\
\noindent\emph{Proof.} (i) $\Rightarrow$ (ii). We prove the contrapositive. Suppose there exist $S,T$, with $1\leq |S|, |T|<d$, $[\PiAcal(S),\PiBcal(T)]=0$. Then also 
$[\PiAcal(S),\PiBcal(T^c)]=0=[\PiAcal(S^c),\PiBcal(T)]=[\PiAcal(S^c),\PiBcal(T^c)]$. So we can assume without loss of generality that  $1\leq |S|\leq |S^c|<d$,  $1\leq |T|\leq |T^c|<d$, and $1\leq |S|\leq |T|<d$. It follows that $|S|+|T|\leq |S|+ |T^c|\leq d$. Suppose there exists $\phi\in\Hcal$ so that $|\psi\rangle=\PiAcal(S)\PiBcal(T)\phi\not=0$. Then $0\not=\psi\in \PiAcal(S)\Hcal \cap\PiBcal(T)\Hcal$, so that $\Acal, \Bcal$ are not COINC. If such a $\phi$ does not exist, then $\PiAcal(S)\Pi_\Bcal(T)=0$. Now let $0\not=\psi\in\PiAcal(S)\Hcal$. Then $|\psi\rangle=\PiAcal(S)\psi=\PiAcal(S)\PiBcal(T^c)\psi=\PiBcal(T^c)\PiAcal(S)\psi$. Hence $0\not=\psi\in \PiAcal(S)\Hcal \cap \PiBcal(T^c)\Hcal$, implying $\Acal, \Bcal$ are not COINC.\\
(ii) $\Rightarrow$ (iii). We prove the contrapositive. Suppose $\mab=0$. Then there exist $i,j$ so that $\langle a_i|b_j\rangle=0$. Possibly renumbering the bases, we can assume $i=1=j$. Hence $|b_1\rangle\in \PiAcal(S)\Hcal$, with $S=\{2,3,\dots, d\}$. Setting $T=\{1\}$, this means $\PiBcal(T)\Hcal\subset \PiAcal(S)\Hcal$. Hence $[\PiAcal(S),\PiBcal(T)]=0$ (Lemma~\ref{lem:commutingproj}~(iii)) and so (ii) does not hold.\\ 
(iii) $\Rightarrow$ (iv). See Eq.~\eqref{eq:mabMab}. 

We now consider the reverse implications. We first show that, if $d=2$ or $3$, then (iii) implies (i) (and hence~(ii) by the above). We will use Lemma~\ref{lem:minorcrit} for that purpose. Note that, if (iii) holds, then $U^\dagger$, which is the inverse of $U$, also has no zero elements. In addition, the matrix elements of $U^\dagger$ are non-zero multiples of $(d-1)$-minors of $U$. This implies that, when $d=2,3$, no minors of $U$ vanish. Hence the bases are COINC.

That (iv) does not imply (iii) is obvious for all $d\geq 3$ and is illustrated by the spin $1$ example given above for $d=3$. 

We now show that (iii) does not imply (ii) if $d\geq 4$. Treating first the case $d=4$,  we consider 
\begin{equation*}
U(s)=\frac12
\begin{pmatrix}
1&1&1&1\\
1&1&-1&-1\\
1&-1&s&-s\\
1&-1&-s&s
\end{pmatrix}, \quad |s|=1; U^\dagger(s)=U(\overline s).
\end{equation*}
These are the most general transition matrices for MUB in dimension $d=4$, up to permutations of rows and columns, and global phase~\cite{Ba22}. Taking
$S=\{1,2\}=T$, we have
$$
\PiAcal(S)=
\begin{pmatrix}
1&0&0&0\\
0&1&0&0\\
0&0&0&0\\
0&0&0&0
\end{pmatrix}
\ 
\textrm{and}
\  
\PiBcal(T)=\frac12
\begin{pmatrix}
1&1&0&0\\
1&1&0&0\\
0&0&1&1\\
0&0&1&1
\end{pmatrix}
$$
These clearly commute so (ii) does not hold. An analogous construction yields the result in all dimensions $d\geq 4$ as we now show. Let $n=d-2$ and set, for some $a,\delta\in\C$,
$$
U=
\begin{pmatrix}
a&a&\delta&\dots&\delta\\
a&a&-\delta&\dots&-\delta\\
\delta&-\delta&u_{11}&\dots&u_{1n}\\
\vdots&\vdots&\vdots&\vdots&\vdots\\
\delta&-\delta&u_{n1}&\dots&u_{nn}
\end{pmatrix}
$$
The normalization and orthogonality of the first two columns and rows implies
$$
|a|=\frac12,\quad |\delta|=\frac1{\sqrt{2n}}.
$$
Now write $u_i=(u_{i1}\dots u_{in})^T\in\R^n$, then the matrix $U$ will be unitary provided one has, for all $1\leq i<j\leq n$,
$$
\langle 1|u_i\rangle=0,\quad \langle u_i|u_i\rangle=1-2\delta^2,\quad \langle u_i|u_j\rangle=-2\delta^2,
$$
where $|1\rangle=(1 \dots 1)^T\in\R^n$. Normalizing the $u_i$, one has $\tilde u_i=u_i/\|u_i\|$ and, for all $1\leq i<j\leq n$, 
$$
\langle \tilde u_i|\tilde u_j\rangle=-\frac1{n-1}.
$$
This means the $\tilde u_i$ form a regular simplex of  $n$ equidistant points on the unit sphere in the $(n-1)$-dimensional Euclidean plane perpendicular to $|1\rangle\in\R^n$.  When $d=4$, $n=2$ and $\tilde u_2=-\tilde u_1$. When $d=5$, $n=3$  and $\tilde u_1, \tilde u_2, \tilde u_3$ form an equilateral triangle. When $d=6$, $n=4$, and $\tilde u_1,\dots \tilde u_4$ form a regular tetrahedron.  The matrix $U$ so constructed clearly satisfies (iii) but, for the same reason as in the case $d=4$, it does not satisfy (ii). 

It remains to show that in dimension $d=4$ and $d=6$ it is not true that (ii) implies (i). 
We start with $d=4$. Let
$$
|c\rangle=\frac1{\sqrt2}(|1-2\rangle +|3+4\rangle),\quad |c'\rangle=\frac1{\sqrt2}(|1-2\rangle-|3+4\rangle),
$$
where 
$$
|1\pm 2\rangle=\frac1{\sqrt2}(|a_1\rangle\pm|a_2\rangle),
$$
and similarly for $|3\pm4\rangle$. Then define, for $\theta,\theta'\in[0,\pi]$,
\begin{eqnarray*}
|b_1\rangle&=& \cos\theta |1+2\rangle +\sin\theta |c\rangle \\
|b_2\rangle&=& -\sin\theta |1+2\rangle +\cos\theta |c\rangle\\
|b_3\rangle&=&\cos\theta' |3-4\rangle +\sin\theta' |c'\rangle \\
|b_4\rangle&=& -\sin\theta' |3-4\rangle +\cos \theta' |c'\rangle
\end{eqnarray*}
The $|b_i\rangle$ clearly form an orthonormal basis for all $\theta, \theta'$. Also, $|1+2\rangle\in\PiBcal(\{1,2\})$ so that $\PiAcal(\{1,2\})\cap \PiBcal(\{1,2\})\not=\{0\}$. Hence the bases $\Acal$ and $\Bcal$ are not COINC. The transition matrix between the two bases is
\begin{eqnarray*}
U=
\begin{pmatrix}
\frac1{\sqrt2}\cos\theta+\frac12\sin\theta&-\frac1{\sqrt2}\sin\theta+\frac12\cos\theta&\frac12\sin\theta'&\frac12\cos\theta'\\
\frac1{\sqrt2}\cos\theta-\frac12\sin\theta&-\frac1{\sqrt2}\sin\theta-\frac12\cos\theta&-\frac12\sin\theta'&-\frac12\cos\theta'\\
\frac12\sin\theta&\frac12\cos\theta&\frac1{\sqrt2}\cos\theta'-\frac12\sin\theta'&-\frac1{\sqrt2}\sin\theta'-\frac12\cos\theta'\\
\frac12\sin\theta&\frac12\cos\theta&-\frac{1}{\sqrt2}\cos\theta'-\frac12\sin\theta'&\frac1{\sqrt2}\sin\theta'-\frac12\cos\theta'
\end{pmatrix}
\end{eqnarray*}
To show (ii) holds, we need to show none of the  commutators $[\PiAcal(S),\PiBcal(T)]$ vanish, if $1<|S|,|T|<d$. Choosing $\theta,\theta'$ so that $\mab>0$, it is immediate from Lemma~\ref{lem:commutingproj}~(v) that, for 
$|S|=1$ or $|T|=1$, $[\PiAcal(S),\PiBcal(T)]\not=0$. Indeed, since the transition matrix has no zeros, no eigenvector $|a_i\rangle$ is contained in any of the spaces $\PiBcal(T)\Hcal$, nor is it perpendicular to it. Since $[\PiAcal(S),\PiBcal(T)]=[\PiAcal(S^{\textrm c}),\PiBcal(T)]=[\PiAcal(S),\PiBcal(T^{\textrm c})]=[\PiAcal(S^{\textrm c}),\PiBcal(T^{\textrm c})]$, the same holds when $|S|=3$ or $|T|=3$.  It remains therefore only to check the cases where
$|S|=2=|T|$, of which there are 36. However, since when the statement is true for some $S,T$, it also is for $S, T^{\textrm c}$, \emph{etc.}, we need to only check  nine cases, which we take to be $S=\{1,2\}, \{1,3\}, \{1,4\}$ and $T=\{1,2\}, \{1,3\}, \{1,4\}$. We will use Lemma~\ref{lem:propnoncommutcrit} and therefore need to check, for each of those choices, that $\PiAcal(S^{\textrm c})\PiBcal(T)\PiAcal(S)\not=0$. If $S=\{i,j\}$, $S^{\textrm c}=\{m,n\}$ and $T=\{k,\ell\}$, we have 
\begin{eqnarray*}
\PiAcal(S^{\textrm c})\PiBcal(T)\PiAcal(S)&=&\\
&\ &\hskip -2cm |a_m\rangle\left[\langle a_m|b_k\rangle\langle b_k|a_i\rangle+\langle a_m|b_\ell\rangle\langle b_\ell|a_i\rangle\right]\langle a_i|+|a_m\rangle\left[\langle a_m|b_k\rangle\langle b_k|a_j\rangle+\langle a_m|b_\ell\rangle\langle b_\ell|a_j\rangle\right]\langle a_j|\\
&\ &\hskip -2cm +|a_n\rangle[\langle a_n|b_k\rangle\langle b_k|a_i\rangle+\langle a_n|b_\ell\rangle\langle b_\ell|a_i\rangle]\langle a_i|+
|a_n\rangle[\langle a_n|b_k\rangle\langle b_k|a_j\rangle+\langle a_n|b_\ell\rangle\langle b_\ell|a_j\rangle]\langle a_j|
\end{eqnarray*}
Since all four terms appearing here are linearly independent, it is enough to check for each choice of $S$ and $T$  that one of the four coefficients in square brackets does not vanish. The computations are straightforward but tedious and we don't reproduce the details here. Setting $\theta=\frac{\pi}{3}=\theta'$, one does indeed find $\PiAcal(S^{\textrm c})\PiBcal(T)\PiAcal(S)\not=0$ for the above choices of $T$ and $S$.

We now turn again to the Tao matrix in $d=6$. We know it is not COINC and will now show it satisfies nevertheless condition~(ii). Using Lemma~\ref{lem:propnoncommutcrit}, one easily sees that the latter is equivalent to 
$$
 \forall S,T\in\llbracket 1,d\rrbracket, 1\leq |S|, |T| <d, U(S^{\textrm{c}}, T)U(S,T)^\dagger\not=0.
 $$
Because of the very larg number of index sets $S,T$ to check, we resorted to a numerical computation to show  that, when $U$ is the Tao matrix, this is indeed true.
\qed

\subsection{The uncertainty diagram of COINC bases}\label{s:unc_coinc}
In this section we analyze what the complete incompatibility of two bases implies for their uncertainty diagram: we will see it can be completely determined and that it has a very simple form. 

\begin{theorem}\label{thm:COINCsupport_new} Let $\Acal, \Bcal$ be two orthonormal bases of $\Hcal$. Then the following statements are equivalent:\\
\begin{tabular}{ll}
(i) & $\Acal$ and $\Bcal$ are COINC;\\
(ii) &$\UNCD(\Acal,\Bcal)\subset \{\na+\nb\geq d+1\}\cap \llbracket 1,d\rrbracket^2$;\\
(iii)& $\UNCD(\Acal,\Bcal)= \{\na+\nb\geq d+1\}\cap \llbracket 1,d\rrbracket^2$
\end{tabular}\\
In particular, $\nabmin=d+1$ and the lower edge $L(\na)$, defined in~\eqref{eq:loweredge}, satisfies $L(\na)=d+1-\na$. 
\end{theorem}

Part~(ii) of the theorem is essentially a rephrasing of the definition of the total incompatibility of two bases in terms of the localization properties of all states in the $\Acal$- and $\Bcal$-representations. Precisely, as already explained in~\cite{SDB21}, it ensures that $\Acal$ and $\Bcal$ are COINC if and only if  all states $|\psi\rangle$ have a support uncertainty that is at least equal to $d+1$. This means, for example, that if $|\psi\rangle$ is the superposition of two (respectively three, \dots) basis vectors $|a_i\rangle$, then it has nonvanishing coefficients on at least $d-1$ (respectively $d-2$, $d-3$,\dots) basis vectors $|b_j\rangle$. 
The uncertainty diagram of COINC bases therefore forms a triangle whose lower edge has equation $\na+\nb=d+1$. The relation $\na(\psi)+\nb(\psi)\geq d+1$ is clearly an uncertainty relation. It is much stronger than the multiplicative uncertainty relation Eq.~\ref{eq:uncprincab}. But of course, it only holds for a restricted class of bases, namely the COINC ones. 

Part~(iii) of the theorem gives a stronger assertion. It says that, if $\Acal$ and $\Bcal$ are COINC, given any two index sets $S,T\subset\llbracket 1,d\rrbracket$ such that $|S|+|T|\geq d+1$, there exist a state $|\psi\rangle$ with 
$S_\psi=S, T_\psi=T$. In fact, the proof shows that the set of such states is an open and dense subset of $\Pi_\Acal(S)\Hcal\cap \Pi_\Bcal(T)\Hcal$, which is a $(|S|+|T|-d)$-dimensional subspace of $\Hcal$. (See Lemma~\ref{lem:COINCsupport} below.)  In particular, a randomly chosen $|\psi\rangle$ in $\Pi_\Acal(S)\Hcal\cap \Pi_\Bcal(T)\Hcal$ has these support properties. The theorem is illustrated graphically in  the third panel of Fig.~\ref{fig:uncdiagTAO} for the example of the DFT with $d=7$, which is COINC as already pointed out.  Note the absence of states with $(n_\Acal(\psi), n_\Bcal(\psi))$ in the region of the $(n_\Acal, n_\Bcal)$-plane above or on the curve $n_\Acal n_\Bcal=d$  and strictly below the line $n_\Acal+n_\Bcal=d+1$. 

The theorem also provides a link between total incompatibility and minimal uncertainty. Suppose $|\psi\rangle$ is a minimal support uncertainty state, for which therefore $\nab(\psi)=\na(\psi)+\nb(\psi)=d+1$.  Suppose now we wish to reduce the uncertainty in an $\Acal$-measurement by considering a state $\phi$ for which $\na(\phi)=\na(\psi)-1$. Then necessarily, for this state, $\nb(\phi)\geq \nb(\psi)+1$. So the gain in precision   (or the decrease of uncertainty) of the $\Acal$-measurement is compensated by at least an equal loss in precision for the $\Bcal$-measurement. 
 In this sense, when the bases are COINC, the increase of precision on the variable associated to $\Acal$ is constrained optimally by the loss of precision on the one associated to $\Bcal$. 
 
\noindent\emph{Proof.}  (i) $\Rightarrow$ (ii).  This is an immediate consequence of the definition of COINC. 
Indeed,  clearly $|\psi\rangle\in \PiAcal(S_\psi)\Hcal\cap \PiBcal(T_\psi)\Hcal$. Hence $|S_\psi|+|T_\psi|>d$. 

(i) $\Leftarrow$ (ii). Let us prove the contrapositive. Suppose $\Acal$ and $\Bcal$ are not COINC. Then there exist $S,T$, with $|S|+|T|\leq d$ and  $\Pi_\Acal(S)\Hcal\cap\Pi_\Bcal(T)\Hcal\not=\{0\}$. Hence there exist a state $|\psi\rangle\in\Pi_\Acal(S)\Hcal\cap\Pi_\Bcal(T)\Hcal$. For this state $n_{\Acal}(\psi)\leq |S|, n_{\Bcal}(\psi)\leq |T|$. Hence $n_{\Acal, \Bcal}(\psi)\leq d$.

Since clearly (iii) implies (ii), it remains to prove (i) implies (iii).
For that purpose, let $1\leq R\leq d$, and consider $S,T\in\llbracket 1,d\rrbracket$ so that $|S|+|T|=d+R$. Reorganizing the labels on the basis vectors we can assume, without loss of generality, that
$$
T=\{1,2,\dots, K\},\quad S=\{1,2, \dots, d-(K-R)\},
$$
with $0\leq K-R\leq d-1$. Note that $R\leq K$. 
 
We need to construct $|\psi\rangle\in\Pi_\Acal(S)\Hcal\cap \Pi_\Bcal(T)\Hcal$ with the condition that $S_\psi=S, T_\psi=T.$ 
Lemma~\ref{lem:COINCsupport}  provides the result.\qed
\begin{lemma}\label{lem:COINCsupport}
If $\Acal$ and $\Bcal$ are COINC, and $S,T\in\llbracket 1,d\rrbracket$ so that $|S|+|T|=d+R$, $R\geq 1$,  then $\Pi_\Acal(S)\Hcal\cap \Pi_\Bcal(T)\Hcal$
is $R$-dimensional and $\Pi_\Acal(S)\Hcal + \Pi_\Bcal(T)\Hcal=\Hcal$. In addition, the subset of  $|\psi\rangle\in \Pi_\Acal(S)\Hcal\cap \Pi_\Bcal(T)\Hcal$ for which 
$T_\psi=T, S_\psi=S$ is dense and open in $ \Pi_\Acal(S)\Hcal\cap \Pi_\Bcal(T)\Hcal$. 
\end{lemma}
\begin{proof}
Note that if $K=R$, then $S=\llbracket 1,d\rrbracket$ and the first statement is immediate.
We therefore consider the case where $R<K$. 
Let  $|\psi\rangle=\sum_{j=1}^K d_j |b_j\rangle\in \Pi_\Bcal(T)\Hcal$, with $d\in\C^K$. 
Then $|\psi\rangle\in \Pi_\Acal(S)\Hcal$ if and only if 
\begin{equation}\label{eq:homd}
\langle a_i|\psi\rangle =0=\sum_{j=1}^K \langle a_i|b_j\rangle d_j, \quad \forall  1\leq d-(K-R)+1\leq i\leq d.
\end{equation}
These are $K-R$ homogeneous equations in the $K$ unknows $d_j$.   
Then, by Lemma~\ref{lem:minorcrit}, the $(K-R)$ by $K$ matrix $(\langle a_i|a_j\rangle)_{ij}$, with $1\leq d-(K-R)+1\leq i\leq d, 1\leq j\leq K$ is of maximal rank $K-R$.  As a result, the solutions of the homogeneous equation~\eqref{eq:homd} form an $R$-dimensional subspace of $\Pi_\Bcal(T)\Hcal$.

By the dimension theorem for the sum of two vector spaces we have
$$
\dim (\PiAcal(S)\Hcal + \PiBcal(T)\Hcal) =\dim \PiAcal(S)\Hcal +\dim \PiBcal(T)\Hcal -\dim (\PiAcal\Hcal(S)\cap\PiBcal(T)\Hcal),
$$
which proves the first statement. 

Now fix $1\leq j\leq K$ and let $T'_j=T\setminus \{j\}$. Then, by the same reasoning $\Pi_\Acal(S)\cap\Pi_\Bcal(T'_j)$ is an $(R-1)$-dimensional subspace of $\Pi_\Acal(S)\cap\Pi_\Bcal(T)$.   It is the subspace of its states for which $\langle b_j|\psi\rangle=0$. Then consider 
$$
\psi\in {\mathcal O}_T:=\left(\Pi_\Acal(S)\right)\cap\Pi_\Bcal(T)\setminus\left(\cup_{j=1}^K \Pi_\Acal(S)\cap\Pi_\Bcal(T'_j)\right).
$$
Then $\langle b_j|\psi\rangle\not=0$, for all $1\leq j\leq K$. So $T_\psi=T$. Note that ${\mathcal O}_T$ is an open dense set because it is obtained by removing from the $R$-dimensional vector space $\Pi_\Acal(S)\cap\Pi_\Bcal(T)$ a finite number of $(R-1)$-dimensional vector spaces. We can similarly define
$$
{\mathcal O}_S:=\left(\Pi_\Acal(S)\right)\cap\Pi_\Bcal(T)\setminus\left(\cup_{i=1}^{d-(K-R)} \Pi_\Acal(S'_i)\cap\Pi_\Bcal(T)\right).
$$
Then $\langle a_i|\psi\rangle\not=0$, for all $1\leq i\leq d-(K-R)$. So $S_\psi=S$. Taking $|\psi\rangle\in {\mathcal O}_T\cap {\mathcal O}_S$, which is still open and dense, we obtain the desired result. 
\end{proof}


\section{Relating complete incompatibility and mutual unbiasedness}\label{s:COINC_MUB}
We have justified our definition of complete incompatibility of two bases $\Acal$ and $\Bcal$ (Definition~\ref{def:COINC}) through an analysis of 
the effect of a $\Bcal$-measurement on previously obtained information in an $\Acal$-measurement. We have seen that, essentially, two bases are COINC if and only if the measurement in $\Bcal$ always perturbs this previous information, whatever the outcome of the $\Acal$-measurement, which can be coarse grained or fine grained. 

One can, alternatively, consider only the situations where the $\Acal$-measurement is fine grained, so that the state after the $\Acal$-measurement is $|a_i\rangle$, and require that the subsequent fine-grained $\Bcal$ measurement has a maximally uncertain outcome, which means that $|\langle a_i|b_j\rangle|=1/\sqrt{d}$, so that all outcomes $b_j$ are equally likely. This idea has lead to the definition of MUB which are, as already mentioned, bases $\Acal$ and $\Bcal$ for which the transition matrix has the property that $\Mab=1/\sqrt d=\mab$. MUB have attracted attention because of their potential use in various quantum information protocols~\cite{Schw60, Iv81, PlRoPe06, DuEnBeZy10}. 
One may note that, when dealing with two conjugate (and hence continuous) variables $X$ and $P$, an analogous property holds; indeed $|\langle x|p\rangle|=\frac{1}{\sqrt{2\pi}}$ in that case, expressing the idea that when the localization in $X$ is perfect, the uncertainty in $P$ is maximal and vice versa. The definition of MUB captures this same phenomenon: perfect localization in the $\Acal$-representation leads to maximal uncertainty in the $\Bcal$-representation. However, conjugate variables have an additional property: their support in the $X$- and $P$-representations cannot be both finite. 
Indeed, it is well known that there do not exist states $|\psi\rangle$ for which the $X$-representation $\langle x|\psi\rangle$ is localized inside some bounded set $S\subset \R$ and the $P$-representation $\langle p|\hat \psi\rangle$ is localized in some bounded set $T\subset\R$~\cite{FoSi97}.
 In fact, if the $X$-support of $|\psi\rangle$ is bounded, its $P$-support must be unbounded.  The definition of complete incompatibility naturally transcribes this second property to a somewhat analogous statement in the finite-dimensional case.  
However, for the DFT in finite dimension, it follows from~\cite{Tao05} and Theorem~\ref{thm:COINCsupport_new} that it is COINC if and only if the dimension is prime. This shows not all MUB are COINC. We point out that, on the contrary, the noncommutativity/incompatibility measures introduced in~\cite{MoKa21} are maximal on MUB and therefore do not distinguish between them.

More generally, one may therefore ask the question under what circumstances MUB are COINC.  It is easy to see that in dimenion $d=2,3,5$ all MUB are COINC, whereas it was shown in~\cite{SDB21} that in dimension $d=4$, none of them are. In fact, we showed above that, in dimension $d=4$, no MUB satisfy the weaker condition (ii) of Proposition~\ref{prop:incompatibilities}. In higher dimension, the situation is not clear, in particular because a complete characterization of all MUB is not known for $d\geq 6$.  

Two questions arise therefore naturally:
\vskip0.2cm

\centerline{``In which dimensions do their exist MUB that are COINC?''}

\noindent and

\centerline{``In which dimensions are all MUB COINC?''}
\vskip0.2cm
\noindent The answer to the first question certainly includes all prime dimensions since, as already mentioned, the DFT is COINC in prime dimension.  Answering this question is probably not easy since a complete characterization of all MUB does not exist. Similarly, the only dimensions eligible for a positive answer to the second question are the prime dimensions. Again, in dimensions $2,3,5$ the assertion is valid, but in higher prime dimensions the answer is not known, to the best of our knowledge. 

While these questions seem difficult, as an immediate consequence of Theorem~\ref{thm:COINC_U_dense}, one nevertheless has that, arbitrarily close to any MUB, there are always bases that are COINC, in any dimension. This is the content of the following theorem, proven in Appendix~\ref{app:coincopendense}. 
\begin{theorem} \label{thm:COINC_mab} For all $d\geq 2$ and for all $0<m<\frac1{\sqrt d}$ there exist bases $\Acal, \Bcal$ that are COINC and whose transition matrix $U$ satisfies 
$m\leq \mab\leq 1/{\sqrt d}$.
\end{theorem}


\section{Support uncertainty as a KD-nonclassicality witness}\label{s:supuncwitness}
We now turn to the connection between the support uncertainty of states $|\psi\rangle$ and their KD-nonclassical nature, as defined in the Introduction. 
Our most general result, under the weakest conditions on the transition matrix $U$ (Theorem~\ref{thm:NCbound}), can be paraphrased as stating that - provided  $U$ does not have too many zeros - states $|\psi\rangle$ with a large support uncertainty are KD nonclassical. In other words, the support uncertainty is a KD-nonclassicality witness. The natural lower bound on  the support uncertainty to obtain this conclusion is, as we shall see, the line $\nab=d+1$, that we shall refer to as the \emph{KD-nonclassicality edge}. Under such general circumstances, the support uncertainty is not a faithful witness of nonclassicality: there may be states with a small support uncertainty that are nevertheless KD nonclassical. An example of this situation can be observed in the central panel of Fig.~\ref{fig:uncdiagTAO}, where the uncertainty diagram of the Tao matrix (which is not COINC, as shown above) is displayed.

When $U$ is COINC and $\mab$ close to $1/\sqrt d$, we obtain a complete characterization of the KD-classical states: we show that the only KD-classical states are the basis states, all others being nonclassical (Theorem~\ref{thm:pert_mubclassical}). 

To state these results, we need some further terminology. Let $\Zc$ be the maximum total number of zeros that can be found  in any two distinct columns of the transition matrix $U$. Let $\Zr$ be the maximum total number of zeros in two distinct rows of $U$. If $Z$ is the total number of zeros, then of course $\Zr, \Zc\leq Z$. If $Z=0,1$ or $2$, then so are $\Zc$ and $\Zr$. But $Z$ can be considerably higher than  $\Zr$ and $\Zc$, as we will see in examples below. Theorem~\ref{thm:NCbound} is an immediate consequence of the following more technical statement, which is of interest in its own right.

\begin{proposition}\label{prop:NCbound_bis} Let $\Acal, \Bcal$ be two orthonormal bases on a $d\geq 2$ dimensional Hilbert space $\Hcal$, with transition matrix $U$. Let $|\psi\rangle\in\Hcal$ and suppose
\begin{eqnarray}
&\ & \max\{\na(\psi), \nb(\psi\}> \max\{\Zr,\Zc\}, \label{eq:zczrbound}\\
 &\ &\nab(\psi)=n_\Acal(\psi)+n_\Bcal(\psi)> d+1,\label{eq:suppunclb}
\end{eqnarray}
then $|\psi\rangle\in\Hcal$ is KD nonclassical.
\end{proposition} 
The proof of this result is given below.

The first condition on $|\psi\rangle$, Eq.~\eqref{eq:zczrbound}, is implied by the second in those cases where there are not ``too many'' zeros in $U$. This idea is made precise in the following theorem.
\begin{theorem}\label{thm:NCbound}
Let $\Acal, \Bcal$ be two orthonormal bases on a Hilbert space $\Hcal$ of dimension $d\geq 2$, with transition matrix $U$.  
Suppose 
$\max\{\Zr,\Zc\}\leq \frac{d+1}{2}$.
Then, if $|\psi\rangle\in\Hcal$ satisfies Eq.~\eqref{eq:suppunclb}, it is KD nonclassical. 
\end{theorem}
This result was proven in~\cite{SDB21} for all $d\geq 2$, but only under the restrictive hypothesis that $U$ has no zeros, in which case Eq.~\eqref{eq:zczrbound} is automatically satisfied for any $|\psi\rangle$. 
The case $d=2$ is particular and fully covered by the following remarks. Note that when $d=2$, then $Z=\max\{\Zr, \Zc\}$. Since a unitary two by two matrix cannot have exactly one zero, it has either one or two zeros. The hypothesis  $\max\{\Zr,\Zc\}\leq \frac{d+1}{2}=3/2$ therefore corresponds to the case where  $U$ has no zeros at all. In that case the theorem, together with the fact that the basis states are classical,  implies that the only nonclassical states are those with $\na(\psi)=2=\nb(\psi)$. If on the contrary $U$ does have two zeros, then it is (equivalent to) the identity matrix, and then there are no nonclassical states at all. 

\noindent{\bf Proof of Theorem~\ref{thm:NCbound}.} We use Proposition~\ref{prop:NCbound_bis}. It is enough to show that when $\max\{\Zr,\Zc\}\leq \frac{d+1}{2}$,  Eq.~\eqref{eq:zczrbound} is implied by Eq.~\eqref{eq:suppunclb}. Suppose therefore that Eq.~\eqref{eq:suppunclb} holds. Then $\max\{\na(\psi), \nb(\psi)\}> \frac{d+1}{2}\geq \max\{\Zr,\Zc\}$, which proves the result. 
\hfill\qed\\

Here are two examples with $d=6$,  that satisfy the hypotheses of the theorem:
\begin{equation}\label{eq:Uzeros}
U_6=\frac{1}{\sqrt5}
\begin{pmatrix}
0&1&1&1&1&1\\
1&0&1&-1&1&-1\\
1&1&0&-1&-1&1\\
1&-1&-1&0&1&1\\
1&1&-1&1&0&-1\\
1&-1&1&1&-1&0
\end{pmatrix},\quad
U'_6=\frac{1}{\sqrt5}
\begin{pmatrix}
0&1&1&1&1&1\\
1&0&1&-1&1&-1\\
1&1&0&1&-1&-1\\
1&-1&1&0&-1&1\\
1&1&-1&-1&0&1\\
1&-1&-1&1&1&0
\end{pmatrix}.
\end{equation}
Indeed, $Z=6, \Zr=2=\Zc$. 
 These examples can be extended to arbitrarily high dimension, as follows. Let us write $|a_i\rangle, |b_j\rangle$ for two bases of $\C^6$ having $U_6$ as transition matrix. Let $\Hcal=\C^6\otimes\C^k$ and let $V$ be the  unitary $k$ by $k$ transition matrix of two bases $|a'_{i'}\rangle, |b'_{j'}\rangle$ of $\C^k$, that we assume to not have any zeros. Hence $U_6\otimes V$ is the transition matrix between the bases $|a_i, a'_{i'}\rangle, |b_i, b'_{j'}\rangle$. 
Then each basis vector $|b_j, b'_{j'}\rangle$ has exactly $k$ zero matrix elements on the basis $|a_i, a'_{i'}\rangle$. So the total number of zeros in the transition matrix is $Z=6k^2$ and any  $2$ such basis vectors have $2k$ vanishing components implying $\Zc=2k$. Similarly $\Zr=2k$. Since the dimension of the Hilbert space is $d=6k$, the condition in Theorem~\ref{thm:NCbound} holds for any $k\in\N_*$. These examples show that the number of zeros in any given column  can in fact be very large, and proportional to the length of the column (which is $6k$ in these examples). 

When the two bases $\Acal$ and $\Bcal$ are mutually unbiased, or in a suitable sense close to mutually unbiased, a stronger result can be proven, that we now turn to.
\begin{theorem}\label{thm:pert_mubclassical}
Let $\Acal, \Bcal$ be two orthonormal bases on a Hilbert space $\Hcal$ of dimension $d\geq 3$, with transition matrix $U$.  Suppose
\begin{equation}\label{eq:mMbound}
\frac{d-1}{d+1}< \left(\frac{\mab}{\Mab}\right)^2\leq 1.
\end{equation}
Suppose $|\psi\rangle$ is KD classical. Then either $|\psi\rangle\in \Acal\cup\Bcal$ or $\na(\psi)+\nb(\psi)\leq d$.  
Consequently, if $\Acal$ and $\Bcal$ are in addition COINC, then the only classical states are the basis states. 
\end{theorem}
We excluded the case $d=2$ from the statement because in that case, if $\mab>0$, we know from Theorem~\ref{thm:NCbound} that, if $|\psi\rangle$ is classical, then $\na(\psi)+\nb(\psi)\leq 3$.  This, in turn, implies that $|\psi\rangle\in\Acal\cup\Bcal$, since $d=2$, proving the result. So the additional constraint on $\mab/\Mab$ is not relevant when $d=2$. 

The theorem implies that for the DFT in prime dimension, which is both mutually unbiased and COINC, only the basis vectors are classical. This can be observed in the rightmost panel of Fig.~\ref{fig:uncdiagTAO} for $d=7$. 

Note that $\mab=1/\sqrt d$ is equivalent to $\mab=\Mab$ and hence to the bases being mutually unbiased. The theorem completely characterizes the classical states of COINC bases that have a large $\mab$, meaning that $\mab$ is close to its maximal possible value $1/\sqrt d$, attained only for MUB. 
More precisely, suppose that, for some $0\leq\delta<1$, 
$$
\mab^2 > \frac1d(1-\delta).
$$
Then the normalization of the columns of $U$ implies that
$$
\mab^2(d-1)+\Mab^2\leq 1,
$$
so that
$$
\left(\frac{\mab}{\Mab}\right)^2 >\frac1{1+\frac{\delta d}{1-\delta}}. 
$$
Hence, provided 
$$
\delta<\frac{2}{d(d-1)}\frac{1}{1+\frac2{d(d-1)}},
$$
hypothesis Eq.~\eqref{eq:mMbound} is satisfied.  This proves Theorem~\ref{thm:COINC_KDNC}~(iv).
Note that, while this is an increasingly restrictive condition as the dimension $d$ grows, we know such $U$ exist by Theorem~\ref{thm:COINC_mab}.

The proof of Theorem~\ref{thm:pert_mubclassical} relies on the arguments in the proof of Proposition~\ref{prop:NCbound_bis} that we therefore prove first.

\noindent\emph{Proof (of Proposition~\ref{prop:NCbound_bis})}  The proof follows the strategy used in~\cite{ArDrHa21} and~\cite{SDB21}. We proceed by contradiction and suppose $|\psi\rangle$ is KD classical and that hypothesis~\eqref{eq:zczrbound} holds. We need to prove that
$$
\na(\psi)+\nb(\psi)\leq d+1.
$$
When $\na(\psi)=1$, or $\nb(\psi)=1$, $|\psi\rangle$ is one of the basis states and the result is then immediate. We can therefore assume that
$\na(\psi),\nb(\psi)\geq 2$. Since the KD distribution (see Eq.~\eqref{eq:KD}) is insensitive to global phase rotations $|a_i\rangle \to \exp(i\phi_i)|a_i\rangle, |b_j\rangle \to \exp(i\phi_j')|b_j\rangle$, we can suppose that all $\langle a_i|\psi\rangle$ and $\langle \psi|b_j\rangle$ are nonnegative (hence real) for $1\leq i, j\leq  d$. Possibly reordering the basis vectors, we can suppose that $\langle a_i|\psi\rangle\not=0\not=\langle b_j|\psi\rangle $ for $1\leq i\leq  n_\Acal(\psi), 1\leq j\leq n_\Bcal(\psi)$ whereas all other $\langle a_i|\psi\rangle$, $\langle b_j|\psi\rangle$ vanish.  By hypothesis, the KD distribution of $|\psi\rangle$ is real and nonnegative. Hence, for the same range of $i$ and $j$, we can  conclude $\langle a_i|b_j\rangle$ is real and nonnegative.

Now, assume $2\leq \nb(\psi)\leq \na(\psi)=d$.   Then the matrix $U$ contains two columns with real nonnegative entries. Since by hypothesis~\eqref{eq:zczrbound}, there are at most $Z_c<d$ zeros in these two columns of $U$ this is in contradiction with the fact that those columns are orthogonal.  

Let us therefore assume that $2\leq n_\Bcal(\psi)\leq n_\Acal(\psi)<d$. 
Then, for $1\leq j< j'\leq n_\Bcal(\psi)$, we have
\begin{equation}\label{eq:vanish}
0=\langle b_j|b_{j'}\rangle=\sum_{i=1}^{n_\Acal(\psi)} \langle b_j|a_i\rangle \langle a_i|b_{j'}\rangle + \sum_{i=n_\Acal(\psi)+1}^{d} \langle b_j|a_i\rangle \langle a_i|b_{j'}\rangle.
\end{equation}
We first show that for all $1\leq j< j'\leq n_\Bcal(\psi)$, one has
\begin{equation}\label{eq:inprod1}
\sum_{i=1}^{n_\Acal(\psi)} \langle b_j|a_i\rangle \langle a_i|b_{j'}\rangle>0.
\end{equation}
To see this, note first that, for any fixed pair $j\not=j'$, the sum in~\eqref{eq:inprod1} contains $n_\Acal(\psi)$ nonnegative terms.  If at least one of those is positive, the sum is positive. But this is the case, since at most $Z_c$ of these terms can vanish, and $Z_c<n_\Acal(\psi)$ by hypothesis~\eqref{eq:zczrbound}. This proves Eq.~\eqref{eq:inprod1}. 

It then follows from Eq.~\eqref{eq:vanish} that, for all $1\leq j< j'\leq n_\Bcal(\psi)$, one has
$$
\sum_{i=n_\Acal(\psi)+1}^{d} \langle b_j|a_i\rangle \langle a_i|b_{j'}\rangle<0.
$$
Defining, for each $1\leq j\leq n_\Bcal(\psi)$ the vector $v_j=(\langle a_{n_\Acal(\psi)+1}|b_{j}\rangle,\dots, \langle a_{d}|b_{j}\rangle)\in\C^{d-n_\Acal(\psi)}$ we see from the above that $\langle v_j|v_{j'}\rangle<0$. It then follows from Lemma~\ref{lem:obtuse1} below that 
$$
n_\Bcal(\psi)\leq d-n_\Acal(\psi)+1.
$$
This proves the result. 

The case where $n_\Acal(\psi)\leq n_\Bcal(\psi)<d$ is treated similarly, inverting the roles of the columns and the rows. 
\hfill\qed

\noindent\emph{Proof (of Theorem~\ref{thm:pert_mubclassical})} 
We proceed as in the proof of Proposition~\ref{prop:NCbound_bis}. 
We can again suppose $2\leq \nb(\psi)\leq\na(\psi)<d$.  Note that this implies that $|\psi\rangle$ is not one of the basis vectors since $\mab>0$. Eq.~\eqref{eq:vanish} now yields
$$
\mab^2\na(\psi)\leq | \sum_{i=n_\Acal(\psi)+1}^{d} \langle b_j|a_i\rangle \langle a_i|b_{j'}\rangle|\leq \Mab^2 (d-\na(\psi)).
$$
Using that $\na(\psi)\leq\nb(\psi)$, this implies
$$
\na(\psi)+\nb(\psi)\leq \frac{2d}{1+\left(\frac{\mab}{\Mab}\right)^2}<\frac{2d}{1+\frac{d-1}{d+1}}=d+1.
$$
It follows that $\na(\psi)+\nb(\psi)\leq d$, which is the desired result. The second statement of the Theorem now follows from
Theorem~\ref{thm:COINCsupport_new}.
\hfill\qed

The following lemma can be understood geometrically as putting an upper bound on the number of vectors in $\C^n$ that can have an obtuse angle between them, two by two.  It is a refinement of a result in~\cite{ArDrHa21}, where only part (i) of the Lemma was proven. It appeared in the Supplementary Material of~\cite{SDB21}, we repeat the proof here for completeness. 
\begin{lemma}\label{lem:obtuse1}
Let $n,k\in\N_*$ and $v_1, v_2, \dots, v_k\in \C^n\setminus\{0\}$. Then the following holds:\\
(i) If $\langle v_i|v_j\rangle\leq 0$ for all $1\leq i<j\leq k$, then $k\leq 2n$.\\
 (ii) If $\langle v_i|v_j\rangle< 0$ for all $1\leq i<j\leq k$, then $k\leq n+1$.
\end{lemma}
The proof follows from an induction argument, given below. For $n=1$, there is no difference between (i) and (ii). Indeed, one may note that one can always take $v_1>0$, by applying a common phase rotation to all $v_i$, which does not change the inner products $\langle v_i| v_j\rangle$ between them. Hence $v_j<0$ for all $j\not=1$. But if $k>2$, then this contradicts the requirement that $v_2v_3\leq 0$. So $k\leq 2$ when $n=1$.  For arbitrary $n$, it is clear the upper bound in (i) is reached by taking for example $v_1=e_1=-v_2, v_3=e_2=-v_4, \dots v_{2n}=-e_n$. Some of the $v_i$ are then orthogonal, so that this set does not satisfy the hypothesis of (ii). When $n=2$, one can  understand geometrically why in (ii) the upper bound is only $k=3$ and not $k=4$, as in (i). For that purpose, let us reason as if we were working in $\R^2$, not $\C^2$. By applying a rotation to all $v_i\in\R^2$, we can consider $v_1=e_1$. Let $v_j=(\cos\theta_j, \sin\theta_j)$. Since we are in the plane, the hypothesis then implies that the angles between $v_j$ and $v_1$  must be larger than $\pi/2$ for all $2\leq j\leq k$, so that $\theta_j\in ]-\pi/2, \pi/2[$. But the hypothesis further implies $|\theta_j-\theta_{j'}|>\frac{\pi}{2}$, which can only be true if $j$ takes at most $2$ values so that $k\leq 3$. \\
\noindent\emph{Proof (of Lemma~\ref{lem:obtuse1})} (i) 
The proof goes by induction. We have seen the result holds for $n=1$. Suppose the result holds for some $n\in\N_*$. We show it holds for $n+1$. Let $v_1, \dots, v_k\in\C^{n+1}$. As above, we can suppose $v_1=a_1e_1$, $a_1>0$. Write $v_j=a_je_1 +w_j$, with $\langle w_j, e_1\rangle=0$, for all $j=2,\dots, k$. By hypothesis, $a_j\leq 0$ for all $2\leq j\leq k$. As a result, for all $2\leq i<j\leq k$,
\begin{equation}\label{eq:induction}
0\geq \langle v_i|v_j\rangle=a_ia_j +\langle w_i|w_j\rangle.
\end{equation}
Hence, for all $2\leq i<j\leq k$, $\langle w_i|w_j\rangle\leq 0$. 
Note that at most one of the $w_j$ can vanish. Indeed, if two of them vanish, say $w_2=0=w_3$, then $a_2\not=0\not=a_3$ and hence  $\langle v_2|v_3\rangle=a_2a_3>0$, which contradicts the hypothesis. We conclude that among the $(k-1)$ vectors $w_j\in\C^n$, there are at least $(k-2)$ that don't vanish and since their mutual inner products are all non-positive, the induction hypothesis allows to conclude that $k-2\leq 2n$ so that $k\leq 2(n+1)$. \\
(ii) The argument is similar and proceeds again by induction. 
This time, by hypothesis, all $a_j<0$, for $j=2,\dots, k$: none of the $a_j$ can vanish. Now suppose one of the $w_j$ vanishes: $w_2=0$. Then, for $j=3,\dots, k$ since $\langle v_1|v_j\rangle =a_1a_j< 0, \langle v_2|v_j\rangle=a_2a_j<0$, which is impossible since $a_1>0, a_2<0$. Consequently, none of the $w_j, j=2, \dots k$ vanishes. We have found therefore $(k-1)$ non vanishing vectors  in $\C^n$ with a negative inner product. So $k-1\leq n+1$ or $k\leq (n+1)+1$.  \hfill\qed 

Theorem~\ref{thm:NCbound} provides an improvement on existing results that we now further analyze. 
First, as already mentioned, the result was proven in~\cite{SDB21} under the much stricter hypothesis that $Z=0$, \emph{i.e.} $\mab>0$. Recall from Proposition~\ref{prop:incompatibilities} that this constitutes a weak incompatibility condition on $\Acal$ and $\Bcal$. 

In~\cite{ArDrHa21}, on the other hand, the following closely related result was obtained, which imposes no condition on the zeros of $U$:
\begin{theorem}\label{thm:arvid} Let $\Acal$ and $\Bcal$ be orthonormal bases in a $d$-dimensional Hilbert space $\Hcal$. Suppose 
$\Mab<1$. Then, if $|\psi\rangle\in\Hcal$ satisfies
\begin{equation}\label{eq:arvidbound}
\nab(\psi)=\na(\psi)+\nb(\psi)>\lfloor 3d/2\rfloor,
\end{equation}
then $|\psi\rangle$ is KD nonclassical.
\end{theorem}
To compare Theorem~\ref{thm:NCbound} and Theorem~\ref{thm:arvid}, we first remark that, in the latter result, the hypothesis on $U$ is essentially empty, so that it has a very general applicability. Indeed, if $\Mab=1$, some of the $|a_i\rangle$ are equal, up to a phase, to some of the $|b_j\rangle$, and then one can split the Hilbert space into a direct sum and study the nonclassicality on the subspace orthogonal to the common basis vectors, as already pointed out above. On the other hand, the conclusion of KD nonclassicality in Theorem~\ref{thm:arvid} is obtained only on states $|\psi\rangle$ for which $\nab(\psi)> \lfloor 3d/2\rfloor$; this is a more restrictive family of states than those satisfying the bound $\nab(\psi)>d+1$ required in Theorem~\ref{thm:NCbound} as soon as $d\geq 4$ and it becomes increasingly restrictive as $d$ grows. In fact is easy to show that, when $d=3$, Theorem~\ref{thm:NCbound} and Theorem~\ref{thm:arvid} are equivalent. To see this, one needs to remark that in that case, the condition on the zeros of the matrix $U$ in Theorem~\ref{thm:NCbound} does not actually constitute a restriction.  
We refer to Appendix~\ref{app:d3} for the details of the argument.

 In the high-dimensional examples above, where $d=6k$, $k\in\N_*$, and $U=U_6\otimes V$, Theorem~\ref{thm:NCbound} guarantees the KD nonclassicality of all states $|\psi\rangle$ for which $\nab(\psi)>6k+1$, whereas Theorem~\ref{thm:arvid} only concludes this when $\nab(\psi)>9k$.


\section{Conclusion}\label{s:conclusion}
  We have analyzed the Kirkwood-Dirac quasi-probability distributions associated with two observables $A$ and $B$, with eigenbases $\Acal$ and $\Bcal$ on a finite dimensional Hilbert space of states $\Hcal$.   We have characterized the Kirkwood-Dirac (non)classical states in terms of their support uncertainty and shown the special role played by the complete incompatibility of the observables in this analysis. In particular, when the observables are both completely incompatible and (close to) mutually unbiased, we have shown that the only Kirkwood-Dirac classical states are the basis vectors of $\Acal$ and $\Bcal$. 
  
A number of questions remain open. Some, relating to the precise link between mutual unbiasedness and complete incompatibility, have been suggested in Section~\ref{s:COINC_MUB}. We have also not attempted to establish a hierarchy among the Kirkwood-Dirac nonclassical states, a question of obvious interest. Finally, the extension of our results to mixed states remains a subject for further research.

\bigskip

\noindent\textbf{Acknowledgements}
This work was supported in part by the Agence Nationale de la Recherche under grant ANR-11-LABX-0007-01 (Labex CEMPI) and by the Nord-Pas de Calais Regional Council and the European Regional Development Fund through the Contrat de Projets \'Etat-R\'egion (CPER).

\appendix

\vskip0.5cm
\centerline{\bf APPENDICES}

\section{Proof of Proposition~\ref{prop:Cinf0}}\label{s:commutator}
To see this, we write 
$
\Hcal=\Hcal_-\oplus \Ker C\oplus \Hcal_+,
$
where $\Hcal_-, \Hcal_+$ are the negative and positive spectral subspaces for $C$ and  $C_-<0, C_+>0$ are the restrictions of $C$ to $\Hcal_-,\Hcal_+$. So
$$
C=C_-\oplus 0\oplus C_+.
$$

Note that, in finite dimension, one always has $\Tr C=0$, as can be seen by using the cyclicity of the trace. So, in finite dimension, either $C=0$, or both $C_-\not=0$ and $C_+\not=0$. 

 In infinite dimension, this argument does not work, since $A$ and/or $B$ may not be trace class. In fact, it turns out that in infinite dimension it is possible that $C_-=0, C_+\not=0$ and Ker$C=\{0\}$. Examples can be found in~\cite{How87, Ka91, HeKr19} where $C$ is trace class and positive. It always has zero in its spectrum then, but not necessarily as an eigenvalue.  In some examples $C$ is finite rank and nonnegative in which case Ker$C\not=\{0\}$. 
 
We now turn to the general proof of~\eqref{eq:Cinf0}.
If $\Ker C\not=\{0\}$ it trivially holds. We therefore suppose  $\Ker C=\{0\}$.
 Let us first consider the case where the commutator  $C$ has both positive and negative spectrum, so that $C_-\not=0\not=C_+$. The proof of~\eqref{eq:Cinf0} then goes as follows. Let $0\not=\psi_\pm\in \Hcal_\pm$ and let 
$|\psi\rangle=\alpha\psi_-+ \beta \psi_+$, with $\alpha,\beta\in \R_*^+$. Then
$$
\parallel \psi \parallel^2=\alpha^2\parallel \psi_-\parallel^2 +\beta^2\parallel \psi_+\parallel^2, \quad\textrm{and}\quad
\langle \psi|C|\psi\rangle=-\alpha^2 \langle \psi_-||C_-| |\psi_-\rangle +\beta^2\langle \psi_+|C_+\psi_+\rangle.
$$
Imposing $\langle \psi|C|\psi\rangle=0, \parallel\psi\parallel=1$, one finds a unique nontrivial solution for  $(\alpha^2, \beta^2)$.  This proves~\eqref{eq:Cinf0}.

Suppose now that $\Ker C=\{0\}, C_-=0, C_+\not=0$. We prove~\eqref{eq:Cinf0} by contradiction. Suppose there exists $\delta>0$ so that 
 $\langle\psi |C|\psi\rangle=\langle\psi |C_+|\psi\rangle\geq \delta$ for all $|\psi\rangle\in\Hcal$. Then, for all $t\in\R^+$, 
$$
\langle \psi, \exp(iAt)B\exp(-iAt)\psi\rangle=\langle \psi, B\psi\rangle+ i\int_0^t \rd s \langle \psi_t, iC\psi_t\rangle \leq \langle \psi, B\psi\rangle- \delta t.
$$
But this is impossible if $B$ is bounded. 
\qed

It follows from this that, in infinite dimension, if $A$ and $B$ are bounded and have a positive commutator $C\geq0$, then $0\in\sigma(C)$ and~\eqref{eq:Cinf0} holds. Indeed, if not $\inf\sigma(C)=\delta>0$. And that leads to the above contradiction.

When both $A$ and $B$ are unbounded, the commutator can of course be positive. The standard case is when they are canonically conjugate so that $C=i\bbone$ and~\eqref{eq:dispunc2} reduces to the Heisenberg uncertainty principle, which provides  a uniform lower bound for all $|\psi\rangle$.

\section{Proof of the support uncertainty principle~\eqref{eq:uncprincab}}\label{s:proof_uncprinc}

Let $|\psi\rangle$ be an arbitrary state in $\Hcal$. Then
\begin{eqnarray*}
|\langle a_i|\psi\rangle|&=&\left|\sum_j \langle a_i|b_j\rangle \langle b_j|\psi\rangle\right|\leq \Mab\sum_j |\langle b_j|\psi\rangle|\\
&\leq&\Mab\sqrt{n_\Bcal(\psi)}\left(\sum_j  |\langle b_j|\psi\rangle|^2\right)^{1/2}\\
&=&\Mab\sqrt{n_\Bcal(\psi)}\left(\sum_k  |\langle a_k|\psi\rangle|^2\right)^{1/2}\\
&\leq&\Mab\sqrt{n_\Bcal(\psi)}\sqrt{n_\Acal(\psi)}\max_k|\langle a_k|\psi\rangle|
\end{eqnarray*}
Taking the maximum over $i$, the result follows. 


\begin{figure*}
\begin{center}
\begin{minipage}[t]{.45\linewidth}
\vspace{0pt}
\centering
\includegraphics[width=5cm]{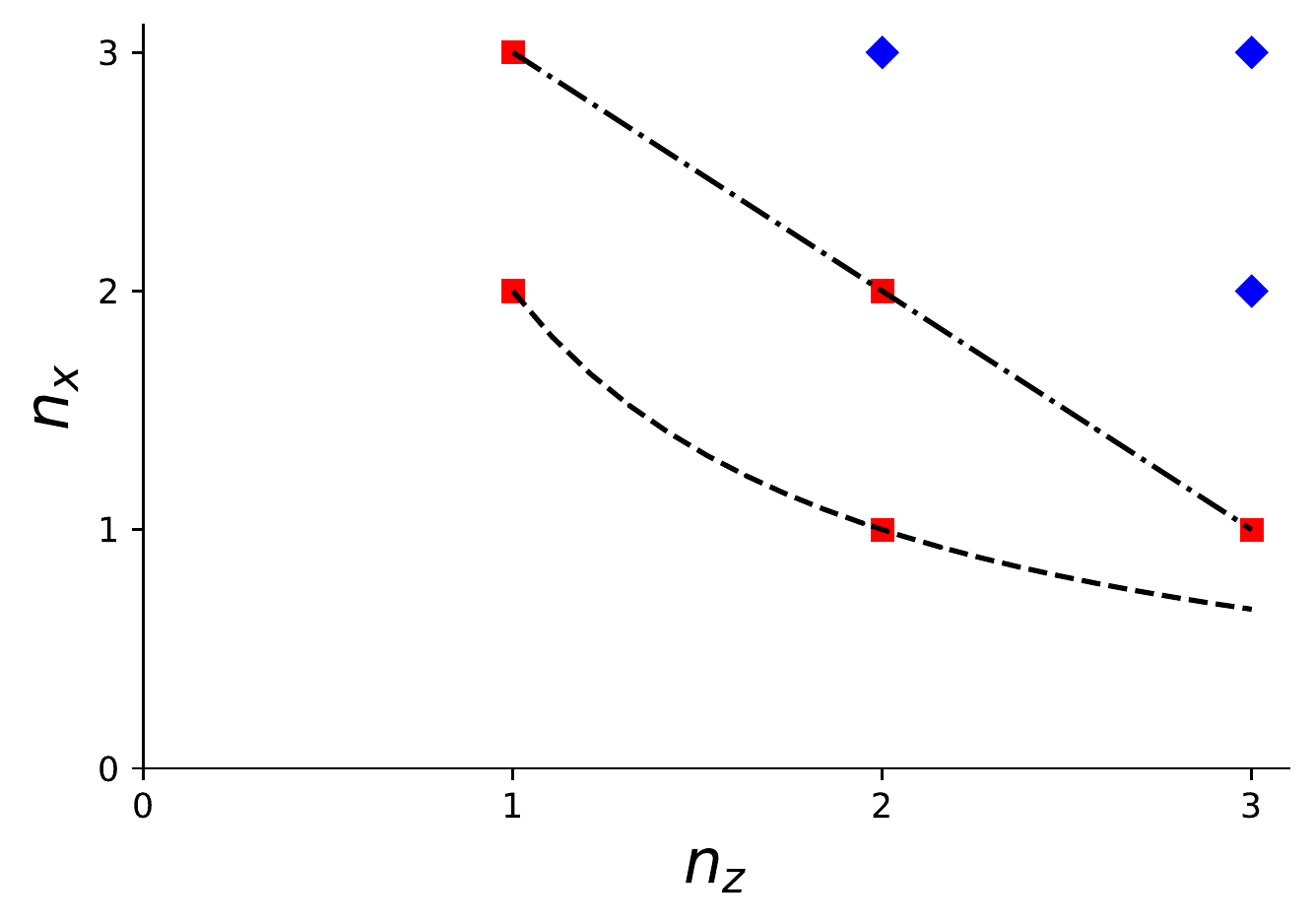}\hspace{1cm}
\end{minipage}%
\begin{minipage}[t]{.45\linewidth}
\vspace{0pt}
\centering
\setlength\extrarowheight{5pt}
\begin{tabular}{|l||*{4}{c|}}\hline
\backslashbox{$T$}{$S$}
&\makebox[4em]{$\{1,0\}$}&\makebox[4em]{$\{1,-1\}$}&\makebox[4em]{$\{0,-1\}$}\\
\hline\hline
 $\{1,0\}$&$\psi_{-1,-1} $ \checked&$\psi_{-1,0}$&$\psi_{-1,1}$\checked\\ \hline
$\{1,-1\}$&$\psi_{0,-1}$ &$\psi_{0,0}$\checked&$\psi_{0,1}$\\ \hline
$\{0,-1\}$ &$\psi_{1,-1}$ \checked&$\psi_{1,0}$& $\psi_{1,1}$\checked  \\ \hline
\end{tabular}
\end{minipage}
\end{center}
\caption{Uncertainty and KD nonclassicality for spin $1$ system. (a) Left panel. The curve $n_{z}n_{x}=2$ (dashed (black) curve) and KD-nonclassicality edge $n_{z}+n_{x}=4$ (dot-dashed (black) line). The (blue) diamonds indicate the $(n_z(\psi), n_x(\psi))$ values corresponding to KD-nonclassical states $\psi$. The (red) squares correspond to KD-classical states.  (b) In the top row and in the first column of the table, one has the possible sets $S,T\subset\{1,0,-1\}$ with exactly two elements. And in each cell of the table, the state generating the intersection $\Pi_z(S)\Hcal \cap \Pi_x(T)\Hcal$, which is one-dimensional.  The state $\psi_{\epsilon, \epsilon'}$ is perpendicular to both $|x,\epsilon\rangle$ and $|z, \epsilon'\rangle$.  The check marks (\checked) indicate that the corrsponding state has $n_{x,z}({\psi_{\epsilon,\epsilon'}})=4$, which corresponds to the KD-nonclassicality edge. Explicit computation of their KD distribution shows they are KD classical. The four other states are equal to either $|x,0\rangle$ or $|z,0\rangle$ and as such are KD classical as well. For them $n_{z,x}({\psi_{\epsilon,\epsilon'}})=3$. \label{fig:spin1}}
\end{figure*}


\section{Spin $1$}\label{app:spinone}

We analyze here the KD nonclassicality of all spin $1$ states. The results are represented graphically in  Fig.~\ref{fig:spin1}. Since $n_z(\psi)$ and $n_x(\psi)$ are integers, the support uncertainty relation rules out only the case $n_z(\psi)=1=n_x(\psi)$, but that is obvious from the matrix $U$ at any rate: the two observables don't have common eigenvectors.  There are exactly two  states on the curve $n_z n_x=2$:
$|z,0\rangle, |x,0\rangle$. They are KD classical, since they are basis vectors. Also, for them $n_{z,x}(\psi)=3$. To study the KD nonclassicality of the other states, remark that, since $Z=1$, it follows from Theorem~\ref{thm:NCbound} that all states with $n_{z,x}(\psi)=n_z(\psi)+n_x(\psi)>4$ are KD nonclassical. This leaves only the states with  $n_{z,x}(\psi)= 4$ to analyze. The four eigenvectors $|z,\pm\rangle, |x,\pm\rangle$ are in this case. They are KD classical, since they are basis vectors. The list of  states $|\psi\rangle$ so that $n_z(\psi)=2=n_x(\psi)$ is given in Fig.~\ref{fig:spin1}. Explicit computations not reproduced here show they are all classical.

In conclusion, in a spin $1$ system all states $|\psi\rangle$ above the nonclassicality edge $n_{z,x}(\psi)=4$ are KD nonclassical. All states on or below the nonclassicality edge are KD classical. Note that the overwhelming majority of states has $n_\psi=6$ and is therefore KD nonclassical.


\section{Commuting and noncommuting projectors}
 We collect here some elementary results on commuting and noncommuting projectors needed in the main part of the article. In what follows, $\Pi$ and $\Pi'$ are two orthogonal projectors on a Hilbert space $\Hcal$. They project onto $\Pi\Hcal$ and $\Pi'\Hcal$ respectively. The following lemma  concerns commuting projectors.
 
 \begin{lemma}\label{lem:commutingproj} 
 Let $\Pi$ and $\Pi'$ be two orthogonal projectors on a Hilbert space $\Hcal$. Then 
 \begin{itemize}
 \item[(i)] $[\Pi,\Pi']=0$ if and only if $\Pi\Pi'$ is the projector onto 
 $\Pi\Hcal\cap\Pi'\Hcal$. 
 \item[(ii)]  $[\Pi,\Pi']=0$ if and only if $(\bbone-\Pi)\Pi'\Pi=0$ if and only if $\Pi'$ leaves $\Pi\Hcal$ invariant.
 \item[(iii)] $\Pi\Hcal\subset\Pi'\Hcal$ if and only if  $\Pi\Pi'=\Pi=\Pi'\Pi$. 
 \item[(iv)] $\Pi\Pi'=0$ if and only if $\Pi\Hcal$ is orthogonal to $\Pi'\Hcal$. 
 \item[(v)] If $\textrm{dim}\,\Pi\Hcal=1$, then $[\Pi,\Pi']=0$ if and only if either $\Pi\Hcal\subset\Pi'\Hcal$ or $\Pi\Hcal$ is orthogonal to $\Pi'\Hcal$.
 \end{itemize}
 \end{lemma}
 \noindent\emph{Proof.} (i) If the commutator vanishes, it is clear that $\Pi\Pi'$ is a projector. In that case, one easily checks that  $\Pi\Pi'\Hcal= \Pi\Hcal\cap\Pi'\Hcal$. Conversely, if $\Pi\Pi'$ is a projector, then it is self-adjoint and therefore $[\Pi,\Pi']=0$. \\
 (ii) This is obvious. \\
 (iii) Suppose $\Pi\Hcal\subset\Pi'\Hcal$. Let $|\psi\rangle\in\Hcal$ and write $\psi=\Pi'\psi+ (\bbone-\Pi')\psi$. Then $\Pi\psi=\Pi\Pi'\psi+\Pi(\bbone-\Pi')\psi$. The second term vanishes since $(\bbone-\Pi')\psi$ is orthogonal to $\Pi'\Hcal$ and therefore to $\Pi\Hcal$.  So $\Pi=\Pi\Pi'$. Taking adjoints, we also have $\Pi=\Pi'\Pi$. The converse is obvious.\\
 (iv) This is obvious.\\
 (v) We know from (i) that the commutator vanishes if and only if $\Pi\Pi'$ projects onto $\Pi\Hcal\cap\Pi'\Hcal$. Also, one either has $\Pi\Hcal\cap\Pi'\Hcal=\Pi\Hcal$ or $\Pi\Hcal\cap\Pi'\Hcal=\{0\}$. In the first case $\Pi\Hcal\subset\Pi'\Hcal$. In the second case, according to (i), $\Pi\Pi'$ is the projector onto $\{0\}$, which means $\Pi\Pi'=0$. But this equivalent to $\Pi\Hcal$ being orthogonal to $\Pi'\Hcal$ by (iv). \qed
 
We now turn to noncommuting projectors. 
\begin{lemma}\label{lem:propnoncommutcrit} Let $\Pi,\Pi'$ be two orthogonal projectors Then $[\Pi, \Pi']\not=0$ if and only if there exists $|\psi\rangle\in\Hcal$ so that $\Pi'\Pi\psi\not\in\Pi\Hcal$, if and only if $(\bbone-\Pi)\Pi'\Pi\not=0$.
\end{lemma}
\noindent\emph{Proof.} This is the contrapositive of point (ii) of Lemma~\ref{lem:commutingproj}. \qed\\


\section{The Tao matrix: uncertainty diagram, nonclassicality}\label{app:tao}
We recall the Tao matrix~\cite{tao03} is the following unitary $6$ by $6$ matrix:
\begin{equation}\label{eq:taomatrix2}
U=\frac1{\sqrt d}\\
\begin{pmatrix}
1&1&1&1&1&1\\
1&1&\omega&\omega&\omega^2&\omega^2\\
1&\omega&1&\omega^2&\omega^2&\omega\\
1&\omega&\omega^2&1&\omega&\omega^2\\
1&\omega^2&\omega^2&\omega&1&\omega\\
1&\omega^2&\omega&\omega^2&\omega&1
\end{pmatrix},
\quad 
\omega=\exp(i\frac{2\pi}{3}).
\end{equation}
Since all matrix elements have modulus $1/\sqrt d$, it can be viewed as the transition matrix between two MUB. We will construct here the corresponding uncertainty diagram. First, note that $U$ is certainly not COINC, since it contains, for example, vanishing 2 minors, as is immediately seen.

We now first determine the uncertainty diagram of $U$, represented in Fig.~\ref{fig:uncdiagTAO}. Note that, if $(\na, \nb)$ belongs to the uncertainty diagram of a transition matrix $U$, then there exist $S, T\subset \llbracket 1,d \rrbracket$  with $|S|=\na, |T|=\nb$ and 
$\Hcal(S,T)\not=\{0\}$. While this necessary condition is not sufficient in general to conclude that $(\na, \nb)$ belongs to the uncertainty diagram, a useful sufficient condition is given in the following lemma. 
\begin{lemma}\label{lem:fullsupport}
Let $S, T\subset \llbracket 1,d \rrbracket$ and suppose $\dim \Hcal(S,T)=L\geq 1$. Suppose that for all $S'\subset S$ for which $|S'|=|S|-1$, one has
$\dim\Hcal(S',T)\leq L-1$, and that for all $T'\subset T$ for which $|T'|=|T|-1$, one has $\dim\Hcal(S,T')\leq L-1$. Then the set of $|\psi\rangle\in\Hcal(S,T)$ for which $\na(\psi)=|S|, \nb(\psi)=|T|$ is an open and dense set in $\Hcal(S,T)$. The opposite implication is also true. 
\end{lemma}
\begin{proof} For each $j\in S$, we have that $\Hcal(S\setminus \{j\}, T)$ is an at most $L-1$-dimensional subspace of $\Hcal(S,T)$. Hence the set $\Hcal(S,T)\setminus \Hcal(S\setminus\{j\}, T)$ is open and dense and this is true also for 
$$
\cap_{j\in S}\Hcal(S,T)\setminus \Hcal(S\setminus\{j\}, T).
$$
 It is composed of all vectors  $|\psi\rangle\in\Hcal(S,T)$  for which $\na(\psi)=|S|$. Similarly, the set of $|\psi\rangle\in\Hcal(S,T)$ for which $\nb(\psi)=|T|$ is open and dense as well. Taking the intersection, the direct implication follows. To prove the opposite implication, one can work by contraposition. Suppose there exist
 $S'\subset S$ for which $|S'|=|S|-1$, and so that 
$\dim\Hcal(S',T)=L$. This of course means that $\Hcal(S',T)=\Hcal(S,T)$. 
Hence all states $|\psi\rangle\in \Hcal(S,T)$ have the property that $S_\psi\subset S'$ so that $\na(\psi)\leq |S'|<|S|$. This concludes the proof. 
\end{proof}
We will use the above lemma repeatedly to determine the uncertainty diagram of the Tao matrix.  Note that, since the matrix is symmetric, so is the uncertainty diagram. We first determine the dimension of $\Hcal(S,T)$ in terms of $|S|$ and $|T|$. 

One may first remark that, 
\begin{equation}\label{eq:STdim}
|S|+|T|\geq 7\quad \Rightarrow\quad\dim\Hcal(S,T)=|S|-[6-|T||]=|S|+|T|-6.
\end{equation}
Indeed, if $|T|$ equals $1$ or $6$, this is immediate. 
Consider now the case where $|T|\geq 2$ and $|S|=5$. Then we are imposing one constraint on the $|T|$ components of $|\psi\rangle\in\Hcal(T)$ on the $\Bcal$-basis. Hence $\dim\Hcal(S,T)=|T|-1=|S|+|T|-6$.  
To proceed, we remark that a numerical computation shows that none of the 400 different $3\times 3$ minors of $U$ vanish. This implies that, when 
$|T|=3$ and $|S|=4$, the two constraints imposed on the three components of $|\psi\rangle=d_{j_1}|b_{j_1}\rangle+ d_{j_2}|b_{j_2}\rangle+d_{j_3}|b_{j_3}\rangle$ are necessarily independent so that $\dim\Hcal(S,T)=3-2=1=|S|+|T|-6$. Now consider the case where $|S|=4=|T|$. One then imposes two linearly independent constraints on $4$ basis vectors. Hence $\dim\Hcal(S,T)=2=|S|+|T|-6$. 
Interchanging the roles of $S$ and $T$ in the previous arguments, \eqref{eq:STdim} follows.  

Next, we have that 
\begin{eqnarray}\label{eq:STdim2}
|S|+|T|=6, |T|\not=2\not=|S|& \Rightarrow&\dim\Hcal(S,T)=0,\\
|S|+|T|=6, |T|=2\ \mathrm{or}\ |S|=2&\Rightarrow &\dim\Hcal(S,T)=0 \ \textrm{or}\ 1.\label{eq:STdim3}
\end{eqnarray}
To see this, first remark that when $|S|=1$ or $|T|=1$, then this is immediate since the basis vectors have full support. When $|S|=3=|T|$, the result follows since all $3\times 3$ minors are nonvanishing. 
Consider the case $|T|=2$ and $|S|=4$. With $T=\{j_1, j_2\}$, one has $|\psi\rangle=d_{j_1}|b_{j_1}\rangle + d_{j_2}|b_{j_2}\rangle$. If $S^c=\{i_1, i_2\}$, then $|\psi\rangle\in \Hcal(S,T)$ provided
$\langle a_{i_1}|\psi\rangle=0=\langle a_{i_2}|\psi\rangle$.   Inspecting the matrix $U$ one observes that two such constraints can be either linearly dependent or independent, depending on the choice of $i_1, i_2, j_1, j_2$. So $\dim\Hcal(S,T)=1$ or $0$.  The same is true with the roles of $S$ and $T$ interchanged, proving~\eqref{eq:STdim2}-\eqref{eq:STdim3}.

Finally, one has
\begin{equation}\label{eq:STdim4}
|S|+|T|\leq 5\Rightarrow \dim\Hcal(S,T)=0.
\end{equation}
Since $\dim\Hcal(S,T)$ is a nondecreasing function of $|S|$ and of $|T|$, this follows from what precedes.

We now use Lemma~\ref{lem:fullsupport} to conclude. First, if $\na+\nb\geq 8$, Eq.~\eqref{eq:STdim} allows us to conclude that $(\na,\nb)$ belongs to the uncertainty diagram. Next, when $\na+\nb\leq 5$, Eq.~\eqref{eq:STdim4} implies $(\na,\nb)$ does not belong to the uncertainty diagram. Consider then the case where $\na+\nb=6$. Clearly $(1,5), (3,3), (5,1)$ do not belong to the uncertainty diagram because for these cases $\dim \Hcal(S,T)=\{0\}$. Since there exist $S, T $ with $|S|=2, |T|=4$ for which $\dim \Hcal(S,T)=1$, it follows from Eq.~\eqref{eq:STdim4} and Lemma~\ref{lem:fullsupport} that $(2,4)$ does belong to the uncertainty diagram. 

It remains to check the case where $\na+\nb=7$. We know $(1,6)$ and $(6,1)$ do belong to the uncertainty diagram. An inspection of the Tao matrix shows that $(2,5)$ and $(5,2)$ do not belong to the uncertainty diagram. For the state 
$$
|\psi\rangle=\omega|b_1\rangle-(1+\omega)|b_2\rangle +|b_3\rangle
$$
one sees readily that $(\na(\psi),\nb(\psi))=(4,3)$. 

In conclusion, the uncertainty diagram of the Tao matrix is includes all $(\na,\nb)$ for which $\na+\nb\geq 7$, except for $(2,5), (5,2)$,
 as well as the points $(2,4)$ and $(4,2)$. 
 In particular, the lower edge $L(\na)$ of the diagram is given by
$$
L(1)=6,\quad L(2)=4,\quad L(3)=4,\quad L(4)=2,\quad L(5)=3,\quad L(6)=1.
$$
 
It remains to discuss the KD (non)classicality of the states in function of their support uncertainty.  We know from Proposition~\ref{prop:NCbound_bis} that all states $|\psi\rangle$ with $\na(\psi)+\nb(\psi)> 7$ are KDNC. These are indicated with (blue) diamonds in Fig.~\ref{fig:uncdiagTAO}.
When  $\na(\psi)+\nb(\psi)\leq 7$, the basis states for which $(\na(\psi),\nb(\psi)=(6,1)$ or $(1,6)$ are all classical. They are indicated as (red) squares in Fig.~\ref{fig:uncdiagTAO}.
We have verified with a numerical computation that, when $(\na(\psi),\nb(\psi))=(3,4)$, $(4,3)$, $(2,4)$ or $(4,2)$, $|\psi\rangle$ is KD nonclassical.


\section{COINCs are an open dense set}\label{app:coincopendense}
We will prove Theorem~\ref{thm:COINC_U_dense} and Theorem~\ref {thm:COINC_mab}  here; they are immediate  consequences of  Theorem~\ref{thm:COINC_dense} below.
The proof relies crucially on Lemma~\ref{lem:minorcrit} that was shown in~\cite{SDB21}. We give, for completeness, a slightly different proof of it, in the following Lemma. 
\begin{lemma}\label{lem:minor_inc} (i) Let $1\leq |S|< d, 1\leq |T|\leq d$. Then \textrm{Rank} $U(S^c,T)<|T|$ if and only if $\PiAcal(S)\Hcal\cap\PiBcal(T)\Hcal\not=\{0\}$. \\
(ii) Let $1\leq |S|< d, 1\leq |T|\leq d$. Then \textrm{Rank} $U(S^c,T)=|T|$ if and only if $\PiAcal(S)\cap\PiBcal(T)=\{0\}$.\\
(iii) Let $1\leq |T|<d$ and let $S$ be such that $|S|=d-|T|$. Then det $U(S^c,T)\not=0$ if and only if $\PiAcal(S)\cap\PiBcal(T)=\{0\}$. \\
(iv) $\Acal$ and $\Bcal$ are COINC if and only if none of the minors of the matrix $U$  vanish.
\end{lemma}
\begin{proof}(i) 
Reordering the bases $\Acal, \Bcal$, we can assume $S^c=\llbracket 1, K\rrbracket$, with $K=|S^c|$ and $T=\llbracket 1, L\rrbracket$ with $L=|T|$. 
Hence
$$
U(S^c,T)=
\begin{pmatrix}
\langle a_1|b_1\rangle&\dots&\langle a_1|b_L\rangle\\
\vdots&\vdots&\vdots\\
\langle a_K|b_1\rangle&\dots&\langle a_K|b_L\rangle
\end{pmatrix}.
$$
Consequently, one has that $|\psi\rangle =\sum_{j=1}^L \beta_j |b_j\rangle\in \PiAcal(S)\Hcal\cap \PiBcal(T)\Hcal$ if and only if 
$\beta=(\beta_1,\dots, \beta_K)^T\in \textrm{Ker}\ U(S^c,T)$. 
\\
\emph{Proof of} $\Rightarrow$. Since 
$$
\textrm{Rank} U(S^c,T) +\textrm{Ker} U(S^c,T)=L=|T|,
$$
it follows from the hypothesis that Ker $U(S^c, T)=\{0\}$, which implies the result.  \\
\emph{Proof of} $\Leftarrow$.  Let $0\not=\psi=\sum_{j=1}^L\beta_j|b_j\rangle\in \PiAcal(S)\Hcal\cap \PiBcal(T)\Hcal$.  Then $U(S^c,T)\beta=0$ so that Ker $U(S^c,T)\not=\{0\}$. 
Hence
$
\textrm{Rank}\ U(S,T) <|T|
$. \\
(ii) This is the contraposition of (i).\\
(iii) This follows directly from (ii) since now the matrix $U(S^c,T)$ is a square matrix. \\
(iv) Note that all minors of $U$ are determinants of matrices of the form $U(S^c,T)$, with $|S^c|=|T|$ and hence $|S|=d-|T|$. Now suppose $\Acal$ and $\Bcal$ are COINC. Then, for any such $S,T$, $\PiAcal(S)\cap\PiBcal(T)=\{0\}$. Hence (iii) implies the determinant of $U(S^c,T)$ does not vanish. It remains to prove the converse. Suppose therefore that, whenever $|S|+|T|=d$, det~$U(S^c,T)\not=0$. Then, according to (iii), $\PiAcal(S)\Hcal\cap \PiBcal(T)\Hcal=\{0\}$. Now, suppose $S'\subset S, T'\subset T$, then we still have  $\PiAcal(S')\Hcal\cap \PiBcal(T')\Hcal=\{0\}$. This concludes the proof. 
\end{proof}

To state Theorem~\ref{thm:COINC_dense}, we need some preliminaries. Let $\Acal=(|a_1\rangle, |a_2\rangle, \dots, |a_d\rangle)$ and $\Bcal=(|b_1\rangle, |b_2\rangle, \dots, |b_d\rangle)$ be two orthonormal bases in $\Hcal$, with transition matrix $U$ and suppose $V, W$ are unitary $d\times d$ matrices. Then we construct the orthonormal bases $\Acal'=(|a'_1\rangle, \dots, |a'_d\rangle)$, $\Bcal'=(|b_1'\rangle, \dots, |b_d'\rangle)$ by
$$
|a_i'\rangle=\sum_j V_{ji}|a_j\rangle, \quad |b'_j\rangle =\sum_k W_{kj}|b_k\rangle.
$$ 
We will write
$
\Acal'=\Acal V,\quad \Bcal'=\Bcal W.
$
Their transition matrix $U'$ is
$$
U'=V^\dagger U W.
$$
Note that, since $|b_i\rangle=\sum_j U_{ji}|a_j\rangle$, we have $\Bcal=\Acal U$. 
Let $\sigma$ be a permutation of $\llbracket 1,d\rrbracket$. 
 Let $P_\sigma=(e_{\sigma_1} \dots e_{\sigma_d})$
where $e_1, \dots, e_d$ is the canonical basis of $\C^d$. Then
$\Acal P_\sigma=(|a_{\sigma_1}\rangle,\dots, |a_{\sigma_d}\rangle)$. 
Remark for further computations that, for all $d'\times d$ matrices $C=(c_1 \dots c_d)$, and all $d\times d'$ matrices $F=
\begin{pmatrix} f_1\\ \vdots \\ f_d\end{pmatrix}$,
$$
P_\sigma P_\tau=P_{(\sigma\circ\tau)}, \quad P_\sigma^T=P_\sigma^\dagger=P_{\sigma^{-1}},\quad CP_\sigma=(c_{\sigma_1} \dots c_{\sigma_d}), \quad P_\sigma F=
\begin{pmatrix} f_{\sigma_1^{-1}} \\ \vdots\\ f_{\sigma_d^{-1}}\end{pmatrix}.
$$
It follows from the above that, if $U$ is the transition matrix for two bases $\Acal$ and $\Bcal$, and if $\Acal'=\Acal P_\sigma, \Bcal'=\Bcal P_\tau$, then the transition matrix $U'$ for $\Acal'$ and $\Bcal'$ is 
\begin{equation}\label{eq:equivU}
U'=P_{\sigma^{-1}}UP_\tau,
\end{equation}
since
\begin{eqnarray*}
U'_{ij}&=&\langle a'_i|b'_j\rangle=\langle a_{\sigma_i}|b_{\tau_j}\rangle=U_{\sigma_i \tau_j}\\
&=&\langle e_{\sigma_i}, U e_{\tau_j}\rangle_{\C^d}=\langle P_{\sigma}e_i, U P_\tau e_j\rangle_{\C^d}\\
&=&\langle e_i, P_\sigma^\dagger U P_\tau e_j\rangle_{\C^d},
\end{eqnarray*}
where $\langle \cdot, \cdot \rangle_{\C^d}$ is the inner product on $\C^d$. Note that $U'$ is obtained from $U$ by permuting its rows by $\sigma$ and its columns by $\tau$. 
The central result of this Appendix is then the following theorem. 
\begin{theorem}\label{thm:COINC_dense} Let $\Acal$, $\Bcal$ be orthonormal bases in $\Hcal$, with transition matrix $U$. Then there exists a family of unitary  matrices $T(\epsilon)$ with $\epsilon>0$, and $\lim_{\epsilon\to0}T(\epsilon)=\bbone_d$, and so that the bases $\Acal(\epsilon)=\Acal T(\epsilon)^\dagger$ and $\Bcal$ are COINC for all sufficiently small $\epsilon$. One has $\Bcal=\Acal(\epsilon)U(\epsilon)$, with $U(\epsilon)=T(\epsilon)U$.
\end{theorem}
\begin{proof}
We will construct $T(\epsilon)$ recursively. Note first that, by Lemma~\ref{lem:minorcrit}, if $U$ has no vanishing minors, then the result holds and we can simply put $T(\epsilon)=\bbone_d$. We therefore suppose that $U$ does have at least one vanishing minor. We will perform a ``downward'' recursion on the number of vanishing minors of $U$. 

Consider therefore an arbitrary unitary $d\times d$ matrix $U$ and let $1\leq L$ be the number of vanishing minors of $U$. Let $M_i$, $i=1,\dots, L$ be the list of the corresponding $\ell_i\times \ell_i$  submatrices of  $U$; $\det M_i=0$. Since $U$ is unitary, we know $\ell_i<d$. We can order the $M_i$ so that $d>\ell_1\geq \ell_2\geq \dots\geq  \ell_L$. Let $\sigma_1$, respectively $\tau_1$, be permutations that send the columns and rows of $M_1$ to $1,2,3,\dots \ell_1$. Then, the $\ell_1$-minor of 
$$
U'=P_{\sigma_1}^\dagger U P_{\tau_1}
$$  
corresponding to its first $\ell_1$ rows and columns vanishes, whereas the $k_1=\ell_1+1$-minor corresponding to its $k_1$ first rows and columns does not. 
Indeed, since $U'$ is obtained from $U$ by permuting its columns and rows, the number of vanishing minors of $U'$ is equal to that of $U$ and the maximal dimension of those minors is still $\ell_1<d$. Writing $M_{1}'^{+}$ for this $k_1\times k_1$ submatrix, we are now in the situation of Lemma~\ref{lem:Mlemma} since, by the above, $\det M_1'^{+}\not =0$. There therefore exists a continuous family $V_{1}(\theta_1)$ of unitary $k_1$ by $k_1$ matrices so that $M_1'^{+}(\theta_1)=V_1(\theta_1)M_{1}'^{+}$ has the property that, there exists $\theta_1^*>0$ so that for all  $0<\theta_1\leq \theta_1^*$, $\det M_1'^{+}(\theta_1)\not=0$ and $\det M_1'(\theta_1)\not=0$. 
We can then construct the continuous family of unitary $d\times d$ matrices
$$
T_1(\theta_1)=
\begin{pmatrix}
V_1(\theta_1)& 0\\
0&\bbone_{d-k_1}
\end{pmatrix},\quad T_1(0)=\bbone_d.
$$
Then $M_1'(\theta_1)$ is the submatrix of $U'(\theta_1)=T_1(\theta_1)P_{\sigma_1}^\dagger U P_{\tau_1}$ containing its first $\ell_1$ columns and rows. Its determinant does not vanish.
Since for small $\theta_1$, $V_1(\theta_1)$ is close to the identity, all non-vanishing minors of $U'$ are non-vanishing minors of $U'(\theta_1)$; therefore the number of nonvanishing minors of $U'(\theta_1)$ has increased by at least one. Hence, the number of vanishing minors of $U'(\theta_1)$ is at most equal to $L-1$ provided $\theta_1$ is small enough and different from $0$. The same is then true for $P_{\sigma_1}U'(\theta_1)P_{\tau_1}^\dagger=P_{\sigma_1}T_1(\theta_1)P_{\sigma_1}^\dagger U$. 

Now, let $0<\epsilon<\theta_1^*$ and set $\theta_1=\epsilon$. Then $U_1(\epsilon):=P_{\sigma_1}T_1(\epsilon)P_{\sigma_1}^\dagger U$ has at most $L_{\epsilon,1}\leq L-1$ vanishing minors. We can then repeat the process. There exists $\theta_2^*$,  $\sigma_2, \tau_2$  and for all $0<\theta_2<\theta_2^*$, $T_2(\theta_2)$, constructed as above, so that $P_{\sigma_2}T_2(\theta_2)P_{\sigma_2}^\dagger U_1(\epsilon)$ has at most $L_{\epsilon,1}-1$ vanishing minors. Note that $\theta_2^*$,  $\sigma_2, \tau_2$ all depend on $\epsilon$: we did not indicate this dependence, not to overly burden the notation. But it is important to keep it in mind. We now define 
$\theta_2(\epsilon)=\min \{\epsilon, \theta_2^*\}$ and $U_2(\epsilon)=P_{\sigma_2}T_2(\theta_2(\epsilon))P_{\sigma_2}^\dagger U_1(\epsilon)$. 
Then $U_2(\epsilon)$ has at most $L_{\epsilon, 2}\leq L-2$ vanishing minors. 

Repeating this process at most $L'_\epsilon\leq L$ times, we conclude that there exist permutations $\sigma_j$ and unitary maps $T_j(\theta_j(\epsilon))$ with $0<\theta_j(\epsilon)\leq \epsilon$ so that, 
$$
T(\epsilon)U:=\Pi_{j=L'_\epsilon}^1 P_{\sigma_j}T_j(\theta_j(\epsilon))P_{\sigma_j}^\dagger U
$$
has no vanishing minors. By construction, $\lim_{\epsilon\to 0} T(\epsilon)=0$. 
Defining $\Acal(\epsilon)=\Acal T^\dagger(\epsilon)$, we have $\Bcal=\Acal U=\Acal (\epsilon)T(\epsilon)U$. Since, by the above, $T(\epsilon)U$ has no vanishing minors, provided $\epsilon$ is small enough, the result follows from Lemma~\ref{lem:minorcrit}. 
\end{proof} 
The central argument of the above proof is contained in the technical Lemma~\ref{lem:Mlemma} below. To state and prove it, we need some further notation. 
Let $M^+=(m_1^+\dots m_k^+)$ be a $k\times k$ matrix, with columns $m_i^+\in\C^k$:
$$
m_i^+=\begin{pmatrix} m_i\\\overline\mu_i\end{pmatrix},  
$$ 
with $m_i\in\C^{\ell}$ with $\ell=k-1$ and $\mu_i\in\C$. We  introduce the $\ell\times \ell$ matrix $M=(m_1 \dots  m_{\ell})$ as well as 
$$
\mu=\begin{pmatrix}\mu_1 \\\vdots \\\mu_{\ell}\end{pmatrix}\in\C^{\ell}.
$$
 So we have
$$
M^+=\left(
\begin{array}{c|c}
M&m_k\\
\hline
\mu^\dagger&\overline \mu_k
\end{array}
\right).
$$
Here, for any $v\in\C^\ell$, $v^\dagger$ is its conjugate transpose. 
We will write Ker$M$ for the null space of $M$, meaning those $x\in\C^{\ell}$ for which $Mx=0$; and $\Ima M$ for the image of $M$, which is its column space. Also, if $v\in\C^{\ell}$ is a vector, $\C v$ is the complex line containing $v$ and $[\C v]^\perp$ is the  vector subspace of $\C^{\ell}$ that is perpendicular to it. If $v\not=0$, this is a $(\ell-1)$-dimensional hyperplane. 
\begin{lemma}\label{lem:Mlemma} Let $k\geq 2$. If $\det M^+\not=0$ and $\det M=0$, then $\mu\not=0$ and
\begin{itemize}
\item[(i)] $\Ker M\cap [\C\mu^\dagger]^\perp=\{0\}$;
\item[(ii)] $\Ker M$ is one-dimensional and $\Ima M$ is $(\ell-1)=(k-2)$-dimensional.
\end{itemize}
In addition, there exists a one-parameter continuous  family of unitary $k$ by $k$ matrices $V(\theta)$, $\theta\in [-\pi, \pi]$, with $V(0)=\bbone_k$ so that, defining $M^+(\theta)=V(\theta)M^+=(m_1^+(\theta)\dots m_k^+(\theta))$, one has $\det M^+(\theta)\not=0\not= \det M(\theta)$, for all small enough
$\theta\not=0$. 
\end{lemma}
Recall that the rank of a matrix is the maximal number of linearly independent columns.  Since $M^+$ is invertible, it has maximal rank $k$. Since the determinant of $M$ vanishes, its rank is strictly less than
$\ell$. The important information in the second statement of the Lemma is therefore that the rank of $M$ equals $\ell-1$. The last statement asserts that a small unitary perturbation of the matrix $M^+$ suffices to make $M$ invertible.

To follow the proof it is convenient to have a low dimensional example in mind. For $k=3$, $\ell=2$, the following matrix $M^+$ satisfies the hypotheses of the theorem
$$
M^+=
\begin{pmatrix}
1&2&1\\
1&2&2\\
3&-1&3
\end{pmatrix}, \quad\mu=\begin{pmatrix} 3\\-1\end{pmatrix},\ [\C\mu^\dagger]^\perp=\begin{pmatrix}1\\3\end{pmatrix},\ \Ker M=\C\begin{pmatrix} 2\\-1\end{pmatrix},\ \Ima M=\C\begin{pmatrix}1\\1\end{pmatrix}.
$$ 
The first two assertions of the Lemma are now clearly satisfied. In the proof, we will construct $V(\theta)$ 
to be the rotation over an angle $\theta$ in the plane spanned by $e_3$ and 
$$
b^+=\begin{pmatrix} b\\ 0\end{pmatrix}=\frac1{\sqrt 2}\begin{pmatrix} 1\\-1\\0\end{pmatrix},
$$
where $b$ is chosen orthogonal to $\Ima M$. It is clear in this example that such rotations will tilt the plane spanned by $m_1^+$ and $m_2^+$ into the plane spanned by $m_1^+(\theta)=V(\theta) m_1^+$ and $m_2^+(\theta)=V(\theta) m_2^+$ in such a way that  their projections $m_1(\theta)$ and $m_2(\theta)$  onto the $e_1-e_2$ plane are no longer aligned. Hence the determinant of $M(\theta)$ does not vanish. This argument works in general, as we now show.

\begin{proof} Note first that, if $\mu=0$, then  $\det M^+=\overline \mu_k\det M$ which is in contradiction with the hypothesis. So $\mu \not=0$. 

We now prove that (i) implies (ii). From the dimension theorem for sums of vector spaces, we know that
$$
\dim (\Ker M + [\C\mu^\dagger]^\perp)=\dim \Ker M + \dim [\C\mu^\dagger]^\perp - \dim (\Ker M \cap [\C\mu^\dagger]^\perp)
$$
Now, $\Ker M + [\C\mu^\dagger]^\perp$ is by definition the linear span of $\Ker M$ and $[\C\mu^\dagger]^\perp$ and is therefore a subspace of $\C^{\ell}$, so that $\dim (\Ker M + [\C\mu^\dagger]^\perp)\leq \ell$. Hence (i) implies that
$$
\dim \Ker M + \dim [\C\mu^\dagger]^\perp\leq \ell.
$$
Since $\mu\not=0$, we know that $\dim [\C\mu^\dagger]^\perp=\ell-1$ and consequently $\dim\Ker M\leq 1$. But by assumption $\det M=0$, so $M$ has a non-trivial kernel. Hence $\dim \Ker M=1$. From the dimension theorem for linear maps, we know that 
$$
\dim \Ker M+\dim \textrm{Im} M=\ell,
$$
so that $\dim\textrm{Im} M=\ell-1$, which proves (ii).

It remains to prove (i). For that purpose, note that 
$$
m_i^+=m_i+\overline{\mu_i} e_k,\quad \textrm{where}\ e_k=\begin{pmatrix} 0\\\vdots\\1\end{pmatrix}\in\C^{k}.
$$
Hence, for all $z\in\C^{\ell}$, one has 
$$
\sum_{i=1}^{\ell} z_i m_i^+=\sum_{i=1}^{\ell}z_i m_i + (\mu^\dagger z)e_k.
$$
Suppose $z\in [\C\mu]^\perp\cap \Ker M$. Then the right hand side vanishes; but since the $m_i^+$ are linearly independent, this implies all $z_i=0$. This proves the result. 

We now turn to the last statement of the Lemma. 

Under the conditions of the Lemma, $[\Ima M]^\perp\subset \C^{\ell}$ is one-dimensional. Let $ b\in\Ima M^\perp$ be a unit vector: $b^\dagger b=1$. Let $b^+=\begin{pmatrix} b\\ 0\end{pmatrix}$ and let $\Pi$ be the orthogonal projector onto the plane spanned by $b^+$ and $e_k$:
$$
\Pi=b^+ b^{+\dagger} +e_k e_k^\dagger\quad \textrm{and}\quad \Pi^\perp=\bbone_k-\Pi.
$$
Note that $\Pi^\perp$ is the orthogonal projector onto the image of $M$ (viewed as a subspace of $\C^{k}$, not of $\C^{\ell}$). Now consider, for all $\theta\in [0,2\pi[$, 
$$
V(\theta)=R(\theta)\Pi +\Pi^\perp, 
$$
where $R(\theta)$ is a rotation in the plane spanned by $b^+$ and $e_k$: 
$$
R(\theta)=\cos\theta \Pi + \sin\theta\left(e_k b^{+\dagger}-b^+e_k^\dagger\right).
$$
One easily checks $V(\theta)$ is unitary, continuous in $\theta$ and that $V(0)=\bbone_k$.

We now consider
$$
M^+(\theta)=V(\theta)M^+=(m_1^+(\theta) \dots m_k^+(\theta)),\quad\textrm{and}\quad
m_i^+(\theta)=\begin{pmatrix}m_i(\theta)\\ \mu_i(\theta)\end{pmatrix}.
$$
Since,  for all $i=1, \dots, \ell$,  $\Pi m_i=0$ by the above construction, a direct computation yields,
$$
m_i^+(\theta)=V(\theta) m_i^+= m_i+ \overline \mu_i R(\theta)e_k=m_i+ \overline \mu_i [\cos\theta e_k -\sin\theta b^+]
$$
and hence
$$
m_i(\theta)= m_i- \overline \mu_i \sin\theta b,\quad \mu_i(\theta)=\mu_i \cos\theta.
$$
It remains to show that $\det  M(\theta)\not=0\not=\det M^+(\theta)$ for all $\theta$ small enough.
To show $\det M(\theta)\not=0$ we need to show the vectors $m_i(\theta)$ are linearly independent. Let $0<\theta<\pi$.  Consider $z_i\in\C$, $i=1,\dots, \ell$ and suppose that
$$
\sum_{i=1}^\ell z_im_i(\theta)=0.
$$
Then 
$$
0=\sum_{i=1}^\ell z_i m_i - (\sum_{i=1}^\ell \overline \mu_i  z_i)\sin\theta b.
$$
Since $b$ is orthogonal to all $m_i$, this implies $z\in \Ker M$ and $z\in [\C\mu^\dagger]^\perp$. Hence by part (i) of the lemma we conclude $z_i=0$ for all $i=1,\dots, \ell$, so that the vectors $m_i(\theta)$ are indeed linearly independent.

Since $\det M^+\not=0$, this remains true, by continuity, for $\det M^+(\theta)$ for $\theta$ small enough.  

\end{proof}


\section{KD nonclassicality in dimension $d=3$}\label{app:d3}
It is a simple matter to check that, when $d=3$, Theorems~\ref{thm:NCbound} and~\ref{thm:arvid} are in fact equivalent. First, note that both theorems require $\Mab<1$. Next, when $d=3$,  $3d/2=4.5$ and hence 
the sufficient condition on $|\psi\rangle$ for KD nonclassicality in Theorem~\ref{thm:arvid} is $n_\Acal(\psi)+n_\Bcal(\psi)\geq 5$. While in Theorem~\ref{thm:NCbound} it is $\na(\psi)+\nb(\psi)>d+1=4$, which is the same condition. 
It remains to show that the additional conditions on the number of zeros of $U$ in Theorem~\ref{thm:NCbound} are in fact always satisfied. 
For that purpose we show that the possible values for $Z$ are $Z=0,1$. 
In other words,  the condition $Z\leq \frac{d+1}{2}=2$ is in fact not restrictive. 
Indeed, suppose $Z\geq 2$. One easily sees that two zeros cannot occur in the same row or column because then necessarily $\Mab=1$.
The matrix is therefore (after possible permutations of columns or rows) of the form 
$$
U=
\begin{pmatrix}
0& \cdot &\cdot\\
\cdot &0 &\cdot\\
x&y&\cdot
\end{pmatrix}.
$$
Unitarity then implies $\overline x y=0$ and hence either $x=0$ or $y=0$, so that there is a third zero, which occurs in one of the first two columns. But this is impossible since $\Mab<1$. 



\begin{thebibliography}{10}

\bibitem{Jo07}
Lars~M. Johansen.
\newblock Quantum theory of successive projective measurements.
\newblock {\em Phys. Rev. A}, 76:012119, Jul 2007.

\bibitem{Sp08}
Robert~W. Spekkens.
\newblock Negativity and contextuality are equivalent notions of
  nonclassicality.
\newblock {\em Phys. Rev. Lett.}, 101:020401, Jul 2008.

\bibitem{HeWo10}
Teiko Heinosaari and Michael~M. Wolf.
\newblock Nondisturbing quantum measurements.
\newblock {\em Journal of Mathematical Physics}, 51(9):092201, 2010.

\bibitem{Pu14}
Matthew~F. Pusey.
\newblock Anomalous weak values are proofs of contextuality.
\newblock {\em Phys. Rev. Lett.}, 113:200401, Nov 2014.

\bibitem{Dr15}
Justin Dressel.
\newblock Weak values as interference phenomena.
\newblock {\em Phys. Rev. A}, 91:032116, Mar 2015.

\bibitem{HeEtAl16}
Teiko Heinosaari, Takayuki Miyadera, and M{\'{a}}rio Ziman.
\newblock An invitation to quantum incompatibility.
\newblock {\em Journal of Physics A: Mathematical and Theoretical},
  49(12):123001, feb 2016.

\bibitem{YuSwDr18}
Nicole Yunger~Halpern, Brian Swingle, and Justin Dressel.
\newblock Quasiprobability behind the out-of-time-ordered correlator.
\newblock {\em Phys. Rev. A}, 97:042105, Apr 2018.

\bibitem{Lo18}
Matteo Lostaglio.
\newblock Quantum fluctuation theorems, contextuality, and work
  quasiprobabilities.
\newblock {\em Phys. Rev. Lett.}, 120:040602, Jan 2018.

\bibitem{CaHeTo20}
Claudio Carmeli, Teiko Heinosaari, and Alessandro Toigo.
\newblock Quantum random access codes and incompatibility of measurements.
\newblock {\em Europhysics Letters}, 130(5):50001, 2020.

\bibitem{UoEtAl21}
Roope Uola, Tristan Kraft, S\'ebastien Designolle, Nikolai Miklin, Armin
  Tavakoli, Juha-Pekka Pellonp\"a\"a, Otfried G\"uhne, and Nicolas Brunner.
\newblock Quantum measurement incompatibility in subspaces.
\newblock {\em Phys. Rev. A}, 103:022203, Feb 2021.

\bibitem{Fe11}
Christopher Ferrie.
\newblock Quasi-probability representations of quantum theory with applications
  to quantum information science.
\newblock {\em Rep. Prog. Phys.}, 74:116001, 2011.

\bibitem{LuBa12}
Jeff~S. Lundeen and Charles Bamber.
\newblock Procedure for direct measurement of general quantum states using weak
  measurement.
\newblock {\em Phys. Rev. Lett.}, 108:070402, Feb 2012.

\bibitem{BaLu14}
Charles Bamber and Jeff~S. Lundeen.
\newblock {Observing {\protect{D}}irac's Classical Phase Space Analog to the
  Quantum State}.
\newblock {\em Phys. Rev. Lett.}, 112:070405, Feb 2014.

\bibitem{ThGiChHoBaLu16}
G.~S. Thekkadath, L.~Giner, Y.~Chalich, M.~J. Horton, J.~Banker, and J.~S.
  Lundeen.
\newblock Direct measurement of the density matrix of a quantum system.
\newblock {\em Phys. Rev. Lett.}, 117:120401, Sep 2016.

\bibitem{ArEtAl20}
David R.~M. Arvidsson-Shukur, Nicole Yunger~Halpern, Hugo~V. Lepage,
  Aleksander~A. Lasek, Crispin H.~W. Barnes, and Seth Lloyd.
\newblock Quantum advantage in postselected metrology.
\newblock {\em Nature Communications}, 11(1):3775, 2020.

\bibitem{ArDrHa21}
David R.~M. Arvidsson-Shukur, Jacob~Chevalier Drori, and Nicole~Yunger Halpern.
\newblock Conditions tighter than noncommutation needed for nonclassicality.
\newblock {\em J. Phys. A}, 54(28):Paper No. 284001, 20, 2021.

\bibitem{DeFaKa19}
S{\'{e}}bastien Designolle, M{\'{a}}t{\'{e}} Farkas, and Jedrzej Kaniewski.
\newblock Incompatibility robustness of quantum measurements: a unified
  framework.
\newblock {\em New Journal of Physics}, 21(11):113053, nov 2019.

\bibitem{MoKa21}
Krzysztof Mordasewicz and Jedrzej Kaniewski.
\newblock Quantifying incompatibility of quantum measurements through
  non-commutativity, 2021.

\bibitem{SDB21}
Stephan De~Bi\`evre.
\newblock Complete incompatibility, support uncertainty, and kirkwood-dirac
  nonclassicality.
\newblock {\em Phys. Rev. Lett.}, 127:190404, Nov 2021.

\bibitem{Schw60}
Julian Schwinger.
\newblock Unitary operator bases.
\newblock {\em Proc. Nat. Acad. Sci. U.S.A.}, 46:570--579, 1960.

\bibitem{Iv81}
I~D Ivonovic.
\newblock Geometrical description of quantal state determination.
\newblock {\em Journal of Physics A: Mathematical and General},
  14(12):3241--3245, dec 1981.

\bibitem{PlRoPe06}
Michel Planat, Haret~C. Rosu, and Serge Perrine.
\newblock A survey of finite algebraic geometrical structures underlying
  mutually unbiased quantum measurements.
\newblock {\em Found. Phys.}, 36(11):1662--1680, 2006.

\bibitem{DuEnBeZy10}
T.~Durt, B.-G. Englert, I~Bengtsson, and K~Zyczkowski.
\newblock On mutually unbiased bases.
\newblock {\em Int. J. Quantum Inf.}, pages 535--640, 2010.

\bibitem{FaKa19}
M\'at\'e Farkas and Jedrzej Kaniewski.
\newblock Self-testing mutually unbiased bases in the prepare-and-measure
  scenario.
\newblock {\em Phys. Rev. A}, 99:032316, Mar 2019.

\bibitem{Ba22}
Teo Banica.
\newblock Complex hadamard matrices and applications.
\newblock {\em arXiv:1910.06911}, 2022.

\bibitem{Ki33}
John~G. Kirkwood.
\newblock Quantum statistics of almost classical assemblies.
\newblock {\em Phys. Rev.}, 44:31--37, Jul 1933.

\bibitem{Di45}
P.~A.~M. Dirac.
\newblock On the analogy between classical and quantum mechanics.
\newblock {\em Rev. Mod. Phys.}, 17:195--199, Apr 1945.

\bibitem{cagl69a}
K.~Cahill and R.~J. Glauber.
\newblock Ordered expansions in boson amplitude operators.
\newblock {\em Phys. Rev.}, 177:1857, 1969.

\bibitem{cagl69b}
K.~Cahill and R.~J. Glauber.
\newblock Density operators and quasi-probability distributions.
\newblock {\em Phys. Rev.}, 177:1882, 1969.

\bibitem{NiChu10}
Michael~A. Nielsen and Isaac~L Chuang.
\newblock {\em Quantum Computation and Quantum Information}.
\newblock Cambridge University Press, 2010.

\bibitem{CTDL15}
Cohen-Tannoudji Claude, Diu Bernard, and Lalo{\"{e}} Franck.
\newblock {\em Quantum Mechanics}, volume~1.
\newblock Wiley, 2015.

\bibitem{Lu51}
G.~L\"uders.
\newblock {\em Ann. Physik}, page 322, 1951.

\bibitem{FoSi97}
Gerald~B. Folland and Alladi Sitaram.
\newblock The uncertainty principle: a mathematical survey.
\newblock {\em J. Fourier Anal. Appl.}, 3(3):207--238, 1997.

\bibitem{WiWi21}
Avi Wigderson and Yuval Wigderson.
\newblock The uncertainty principle: variations on a theme.
\newblock {\em Bull. Amer. Math. Soc. (N.S.)}, 58(2):225--261, 2021.

\bibitem{We50}
Hermann Weyl.
\newblock {\em The Theory of Groups and Quantum Mechanics}.
\newblock Dover Publications INC, 1950.

\bibitem{DoSta89}
David~L. Donoho and Philip~B. Stark.
\newblock Uncertainty principles and signal recovery.
\newblock {\em SIAM J. Appl. Math.}, 49(3):906--931, 1989.

\bibitem{GhoJa11}
Saifallah Ghobber and Philippe Jaming.
\newblock On uncertainty principles in the finite dimensional setting.
\newblock {\em Linear Algebra Appl.}, 435(4):751--768, 2011.

\bibitem{Tao05}
Terence Tao.
\newblock An uncertainty principle for cyclic groups of prime order.
\newblock {\em Math. Res. Lett.}, 12(1):121--127, 2005.

\bibitem{Rob29}
H.~P. Robertson.
\newblock The uncertainty principle.
\newblock {\em Phys. Rev.}, 34:163--164, Jul 1929.

\bibitem{Coles17}
Patrick~J. Coles, Mario Berta, Marco Tomamichel, and Stephanie Wehner.
\newblock Entropic uncertainty relations and their applications.
\newblock {\em Rev. Mod. Phys.}, 89:015002, Feb 2017.

\bibitem{Cycon87}
H.~L. Cycon, R.~G. Froese, W.~Kirsch, and B.~Simon.
\newblock {\em Schr\"{o}dinger operators with application to quantum mechanics
  and global geometry}.
\newblock Texts and Monographs in Physics. Springer-Verlag, Berlin, study
  edition, 1987.

\bibitem{How87}
James~S. Howland.
\newblock Perturbation theory of dense point spectra.
\newblock {\em J. Funct. Anal.}, 74(1):52--80, 1987.

\bibitem{Ka91}
Tosio Kato.
\newblock Positive commutators {$i[f(P),g(Q)]$}.
\newblock {\em J. Funct. Anal.}, 96(1):117--129, 1991.

\bibitem{HeKr19}
Ira Herbst and Thomas~L. Kriete.
\newblock The {H}owland-{K}ato commutator problem.
\newblock In {\em Analysis and operator theory}, volume 146 of {\em Springer
  Optim. Appl.}, pages 191--223. Springer, Cham, 2019.

\bibitem{De83}
David Deutsch.
\newblock Uncertainty in quantum measurements.
\newblock {\em Phys. Rev. Lett.}, 50:631--633, Feb 1983.

\bibitem{MaUf88}
Hans Maassen and J.~B.~M. Uffink.
\newblock Generalized entropic uncertainty relations.
\newblock {\em Phys. Rev. Lett.}, 60:1103--1106, Mar 1988.

\bibitem{AbEtAl15}
Kais Abdelkhalek, Ren\'e Schwonnek, Hans Maassen, Fabian Furrer, J\"org Duhme,
  Philippe Raynal, Berthold-Georg Englert, and Reinhard~F. Werner.
\newblock Optimality of entropic uncertainty relations.
\newblock {\em International Journal of Quantum Information}, 13(06):1550045,
  2015.

\bibitem{hi89}
Mark Hillery.
\newblock Total noise and nonclassical states.
\newblock {\em Phys. Rev. A}, 39:2994--3002, Mar 1989.

\bibitem{Tao04}
Terence Tao.
\newblock Fuglede's conjecture is false in 5 and higher dimensions.
\newblock {\em Math. Res. Lett.}, 11(2-3):251--258, 2004.

\bibitem{footnote1}
This last statement is false when $A$ and $B$ have degenerate spectra. To see
  this take $A=\bbone=B$ and choose two bases for which $\Mab<1$.

\bibitem{tao03}
Terence Tao.
\newblock Fuglede's conjecture is false in 5 and higher dimensions.
\newblock {\em Math. Res. Lett.}, 11(2-3):251--258, 2004.

\end{thebibliography}
%

\end{document}